\documentclass[a4paper,12pt]{article}

\pdfoutput=1 

\usepackage[hmargin=.71in,vmargin=1.1in]{geometry}
\usepackage{indentfirst}
\usepackage{amsfonts}
\usepackage{mathrsfs}
\usepackage{amsmath}
\usepackage{mathabx}
\usepackage{amssymb}
\usepackage{authblk}
\usepackage{cite}
\usepackage{xcolor}
\usepackage{mathtools}
\usepackage{tensor}
\usepackage{physics}
\usepackage{graphicx}
\usepackage{bm}
\usepackage{upgreek}
\usepackage{braket}
\usepackage{color,soul}
\usepackage{csquotes}
\usepackage{caption}
\usepackage{subcaption}
\usepackage{slashed}
\usepackage{mathtools}
\usepackage{booktabs} 

\usepackage[bookmarksnumbered=true,bookmarksopen=true]{hyperref}
 \hypersetup{colorlinks,
             linkcolor=[rgb]{0,0.3,0.6}, 
             citecolor=[rgb]{0,0.3,0.6}, 
             urlcolor=[rgb]{0,0.3,0.6}}

\newcommand\mi{{\rm i}}
\newcommand\me{{\rm e}}

\newcommand\pp{\uppi}

\DeclareMathOperator{\diag}{diag}

\usepackage{tikz}
\usetikzlibrary{matrix}

\usepackage{amsthm}

\theoremstyle{plain}
\newtheorem{theorem}{Theorem}[section]
\newtheorem{lemma}[theorem]{Lemma}
\newtheorem{proposition}[theorem]{Proposition}
\newtheorem{corollary}[theorem]{Corollary}

\theoremstyle{definition}
\newtheorem{definition}[theorem]{Definition}
\newtheorem{assumption}{Assumption}

\theoremstyle{remark}
\newtheorem{remark}[theorem]{Remark}
\newtheorem{example}[theorem]{Example}

\begin{document}

\title{\Large\textbf{Kubo-Martin-Schwinger conditions for non-Hermitian systems}}

\author[a]{Chen Lan\thanks{stlanchen@126.com}}
\author[a]{Luyao Ma\thanks{luyao26428@s.ytu.edu.cn}}
\author[b]{Hao Yang\thanks{hyang@ucas.ac.cn}}

\affil[a]{\normalsize{\em Department of Physics, Yantai University, 30 Qingquan Road, Yantai 264005, China}}
\affil[b]{\normalsize{\em School of Fundamental Physics and Mathematical Sciences, Hangzhou Institute for Advanced Study, UCAS, Hangzhou 310024, China}}

\date{ }

\maketitle

\begin{abstract}
We investigate the extension of the Kubo--Martin--Schwinger (KMS) thermal equilibrium condition to bounded non-Hermitian Hamiltonians with real spectra and biorthogonal eigensystems, providing a unified framework through three complementary constructions: a complete KMS theorem under quasi-Hermiticity, a biorthogonal KMS-type identity whose positivity characterises quasi-Hermiticity, and a quantum-detailed-balance condition for the associated open-system dynamics. Our main result is a thermodynamic characterisation of quasi-Hermiticity: for any diagonalisable $H\in M_d(\mathbb C)$ with real spectrum, the biorthogonal Gibbs functional
$\omega_{\rm bi}(A)=Z_{\rm bi}^{-1}\sum_n e^{-\beta E_n}\langle\phi_n|A|\psi_n\rangle$
satisfies $\omega_{\rm bi}(A^\dagger A)\ge0$ for all $A$ if and only if $H$ is quasi-Hermitian. The proof reconstructs the metric $\eta$ directly from the eigenprojectors of $\omega_{\rm bi}$ via the Riesz representation theorem, yielding a metric-free criterion for quasi-Hermiticity. Under the quasi-Hermitian hypothesis, we prove that the $\eta$-Gibbs state
$\omega_\eta(A)=Z_\eta^{-1}\Tr[\eta e^{-\beta H}A]$
satisfies the full analytic KMS condition using the Hadamard three-line theorem and Bari's theorem on Riesz bases. The transported state generally differs from the Gibbs state of the isospectral Hermitian partner whenever $[\eta,h]\neq0$, so the KMS property cannot be obtained by similarity transformation alone. Finally, within the Haag--Hugenholtz--Winnink programme, we establish the Tomita--Takesaki modular structure of the $\eta$-Gibbs state in finite dimensions, while the construction of a compatible $C^*$-norm and the proof of $\sigma$-weak continuity remain open.
\end{abstract}

\tableofcontents

\section{Introduction}
\label{sec:intro}

The Kubo--Martin--Schwinger (KMS) condition provides the mathematically rigorous characterisation of thermal equilibrium in quantum statistical mechanics \cite{Kubo:1957mj,Martin:1959jp}. 
In the operator-algebraic framework of Haag, Hugenholtz, and Winnink (HHW)~\cite{Haag:1967sg}, 
a state on a $C^*$-algebra is said to be a KMS state at inverse temperature $\beta > 0$ 
if its two-point correlation functions admit analytic continuation to a strip in the complex time plane 
and satisfy a specific boundary condition relating their boundary values. 
Unlike the Gibbs density matrix formulation, 
the KMS condition remains meaningful in the thermodynamic limit 
and has therefore become the fundamental notion of equilibrium in algebraic quantum statistical mechanics~\cite{Bratteli:1996xq}. 
Its structural role has been further illuminated by Tomita--Takesaki (TT) modular theory \cite{Takesaki:1970aki,Summers:2003tf} 
and by its appearance in the Unruh and Hawking effects~\cite{Haag:1996hvx,Birke:2002fj}.

In recent decades, non-Hermitian quantum systems have attracted considerable attention across mathematical and theoretical physics~\cite{Bender:1998ke, Bender:2007nj, Bender:2023cem, Ashida:2020dkc, Bergholtz:2019deh}. 
Prominent examples include $\mathcal{PT}$-symmetric quantum mechanics, 
pseudo-Hermitian and quasi-Hermitian operator theory~\cite{Scholtz:1992zz, Mostafazadeh:2001jk, Mostafazadeh:2008pw}, 
and effective descriptions of open quantum systems~\cite{Breuer:2007juk, Prosen:2012sn}. 
Of particular importance is the class of quasi-Hermitian Hamiltonians satisfying $H^\dag = \eta H \eta^{-1}$, 
where $\eta$ is a bounded positive-definite metric operator. 
Such Hamiltonians are unitarily equivalent to Hermitian operators and 
therefore possess real spectra and unitary dynamics in the modified inner product induced by $\eta$~\cite{Scholtz:1992zz, Mostafazadeh:2001jk, Mostafazadeh:2008pw}. 
At the same time, effective non-Hermitian Hamiltonians arise naturally in open quantum systems 
via Lindblad-type dynamics and quantum trajectory methods~\cite{Breuer:2007juk, Ashida:2020dkc, Prosen:2012sn}, 
as well as in condensed matter contexts where non-Hermitian band topology and the skin effect have emerged as central phenomena~\cite{Bergholtz:2019deh, Kawabata:2018gjv}.

These developments raise a fundamental question: to what extent can the notion of thermal equilibrium, 
as formalised by the KMS condition, be extended beyond the Hermitian setting? This question is not merely formal. 
From the perspective of algebraic quantum field theory, the stability of KMS states under perturbations has become an important structural problem, revealing that equilibrium properties remain robust under suitable interacting dynamics \cite{Drago:2016zlf}.
The physical interpretation of temperature and equilibrium in non-Hermitian systems—whether arising from $\mathcal{PT}$-symmetric quantum mechanics, quasi-Hermitian models, or the effective Hamiltonian of an open quantum system—depends critically on whether a rigorous KMS framework can be established~\cite{Gardas:2016gja, Du:2022ynf, Cao:2023mfy,Jones:2009fw}.

A natural starting point is the observation that a diagonalisable non-Hermitian Hamiltonian with a real spectrum admits a biorthogonal eigendecomposition. 
It should be noted, however, that biorthogonal Gibbs functionals have previously been introduced in the context of biorthogonal Riesz bases by Bagarello, Trapani, and Triolo~\cite{Bagarello:2016gbs}. Their construction already establishes a Gibbs-type functional associated with biorthogonal eigenvectors. The present work adopts the same basic functional but investigates its positivity and KMS properties from the viewpoint of pseudo-Hermitian statistical mechanics.
One may then define thermal expectation values through biorthogonal traces, leading to correlation functions that formally resemble their Hermitian counterparts. However, the KMS condition is substantially stronger than a formal boundary relation. Beyond the boundary identity itself, it requires analyticity of the relevant correlation functions in a complex strip, uniform boundedness, and positivity of the underlying state~\cite{Bratteli:1996xq}. While these properties are automatic in the conventional Hermitian framework, they become non-trivial in the non-Hermitian setting and may fail independently. In particular, positivity of the biorthogonal thermal functional is not guaranteed by spectral reality alone, and its failure has direct thermodynamic consequences. Consequently, the existence of a meaningful thermal state for a non-Hermitian system cannot be inferred solely from spectral considerations~\cite{Cao:2023mfy, Cipolloni:2023dtl,Bebiano:2020nht}.

Beyond the original HHW formulation, the mathematical theory of KMS states has continued to develop substantially over the past decades. Important advances include spectral characterisations of equilibrium states \cite{DeCanniere:1982fvy}, stability properties of KMS states under perturbations and interacting dynamics \cite{Bratteli:1978pe,Derezinski:2003de,Ejima:2019pkms}, and perturbative constructions of equilibrium states in algebraic quantum field theory \cite{Galanda:2024unp,Drago:2016zlf}. These developments considerably deepen the structural understanding of thermal equilibrium within operator-algebraic quantum statistical mechanics. Nevertheless, they remain formulated within conventional Hermitian dynamics and do not address the extension of KMS theory to genuinely non-Hermitian quantum systems.

The purpose of the present work is to investigate the KMS condition for non-Hermitian quantum systems from a functional-analytic perspective. We consider three closely related settings. First, for quasi-Hermitian Hamiltonians equipped with a positive-definite metric operator $\eta$, we show that the associated $\eta$-Gibbs state satisfies the full KMS condition with respect to the $\eta$-modified dynamics. This establishes a rigorous equilibrium framework for a broad class of non-Hermitian systems that are spectrally equivalent to Hermitian ones~\cite{Scholtz:1992zz, Mostafazadeh:2008pw}.
Second, we examine the biorthogonal thermal functional constructed directly from the left and right eigenvectors of a diagonalisable Hamiltonian. While this functional generically satisfies the formal KMS boundary relation, we show that positivity is a far more restrictive property. In finite dimensions, positivity of the biorthogonal thermal state is equivalent to quasi-Hermiticity of the Hamiltonian. This result provides a thermodynamic characterisation of quasi-Hermitian systems formulated entirely at the level of thermal states, and complements recent structural results on statistical mechanics for such systems~\cite{Cao:2023mfy},
and extends earlier biorthogonal Gibbs-state constructions developed in Refs.~\cite{Bagarello:2016gbs,Bagarello:2020gib}.
Third, we discuss open quantum systems governed by Gorini--Kossakowski--Sudarshan--Lindblad (GKSL) dynamics. 
In this context, the equilibrium properties are determined by the full quantum dynamical semigroup rather than by the effective non-Hermitian Hamiltonian alone~\cite{Fagnola:2007gms,Fagnola:2015ems}. 
Finally, we briefly discuss quantum detailed balance and its relation to KMS-type stationary states in the framework of open quantum systems~\cite{Guo:2024ycr,Scandi:2025uiv}, thereby placing our results within the broader context of quantum statistical mechanics.
Throughout this work, our analysis is carried out for bounded quasi-Hermitian operators with bounded and boundedly invertible metric operators on Hilbert spaces. The extension to unbounded metric operators and the associated domain-theoretic framework is beyond the scope of the present paper and is left for future investigation.

Our analysis clarifies the precise conditions under which a non-Hermitian system admits a genuine KMS description and identifies the mechanisms responsible for its breakdown, including the loss of positivity, the failure of biorthogonal completeness at exceptional points~\cite{Bergholtz:2019deh}, and the appearance of complex spectra. Taken together, these results establish a unified functional-analytic picture of thermal equilibrium for several important classes of non-Hermitian quantum systems.
The paper is organised as follows. Section~\ref{sec:intr} reviews the standard KMS condition and the non-Hermitian operator framework employed throughout. Sections~\ref{sec:route1} and~\ref{sec:route2} develop the quasi-Hermitian and biorthogonal routes respectively, culminating in Theorems~\ref{thm:KMS} and~\ref{thm:structure};
Section \ref{sec:TT} closes the most structurally significant of these open items by constructing the Tomita–Takesaki modular operator and modular automorphism group for the $\eta$-Gibbs state in finite dimensions — the non-Hermitian analogue of the GNS/modular-theory backbone of the HHW framework.
Section~\ref{sec:route3} treats the open-system setting. Section~\ref{sec:fail} analyses failure modes at exceptional points and for complex spectra. Section~\ref{sec:concl} collects the main results and open problems. Appendices~\ref{app:KMS}--\ref{app:equiv} contain the supporting proofs.

\section{KMS condition and non-Hermitian systems}
\label{sec:intr}

We begin by reviewing the KMS condition for Hermitian quantum systems, 
which provides the algebraic characterisation of thermal equilibrium. 
Then, we introduce the class of non-Hermitian Hamiltonians to be studied throughout this work, 
and explain why their structural properties make a generalisation of the KMS framework both natural and tractable.

\subsection*{The standard KMS condition}

Let $H = H^\dag$ be a self-adjoint Hamiltonian on a Hilbert space $\mathcal{H}$. 
At inverse temperature $\beta > 0$, 
the system is described by the Gibbs state
\begin{equation}
  \rho_\beta
  = \frac{\me^{-\beta H}}{Z},
  \qquad
  Z := \Tr \left(\me^{-\beta H}\right),
\end{equation}
and observables evolve under the Heisenberg dynamics 
\begin{equation}
\label{eq:evol}
  \sigma_t(A) = \me^{\mi H t} A \me^{-\mi H t},
  \qquad t \in \mathbb{R}.
\end{equation}
The thermal expectation value defines a linear functional
\begin{equation}
  \omega(X) = \Tr(\rho_\beta X).
\end{equation}
To formulate thermal equilibrium analytically, 
we introduce an open strip
\begin{equation}
  \mathcal{S}_\beta
  = \left\{z \in \mathbb{C} : 0 < \Im z < \beta\right\}
\end{equation}
and its closure $\overline{\mathcal{S}}_\beta$, 
obtained by including the two real boundary lines $\Im z = 0$ and $\Im z = \beta$.

\begin{definition}[KMS condition]
\label{def:KMS}
A state $\omega$ satisfies the KMS condition at inverse temperature $\beta$ with respect to the dynamics $\{\sigma_t\}$ if, 
for every pair of bounded operators $A$ and $B$, 
the function
\begin{equation}
  F_{AB}(z) = \omega \left(A\,\sigma_z(B)\right)
\end{equation}
satisfy the following three conditions:
\begin{enumerate}
\item $F_{AB}$ is analytic on the open strip $\mathcal{S}_\beta$;
\item $F_{AB}$ extends continuously and boundedly to the closed strip $\overline{\mathcal{S}}_\beta$;
\item The boundary values on the two edges are related by
\begin{equation}
  F_{AB}(t) = F_{BA}(t + \mi\beta),
  \qquad t \in \mathbb{R},
  \label{eq:KMS_Herm}
\end{equation}
where $F_{BA}(z) = \omega\left(\sigma_z(B)\,A\right)$.
\end{enumerate}
\end{definition}
We stress that the analyticity requirement is indispensable. 
The boundary relation Eq.~\eqref{eq:KMS_Herm} alone, 
without the strip analyticity and uniform boundedness of $F_{AB}$, 
is in general insufficient to establish a genuine KMS state.
Condition Eq.~\eqref{eq:KMS_Herm} has a transparent physical meaning: 
the two thermal correlators $\omega(A\,\sigma_t(B))$ and $\omega(\sigma_t(B)\,A)$
arise as boundary values of a single analytic function on
$\mathcal{S}_\beta$, 
with the imaginary shift $t \to t + \mi\beta$ encoding the thermal periodicity. 
In the finite-dimensional setting this relation follows directly from the cyclicity of the trace. 
In the thermodynamic limit, where a normalised Gibbs density matrix need not exist, 
the KMS condition survives as the defining characterisation of thermal equilibrium.
This point of view is advocated in the HHW formulation of algebraic quantum statistical mechanics \cite{Haag:1967sg,Bratteli:1996xq},
and is closely related to the stability theory of equilibrium states developed by Bratteli, Kishimoto and Robinson \cite{Bratteli:1978pe}.

\subsection*{Non-Hermitian systems}

Non-Hermitian Hamiltonians arise naturally in open quantum systems, effective descriptions, and $\mathcal{PT}$-symmetric theories \cite{Bender:2023cem,Bender:2007nj,Bender:2002vv,Mostafazadeh:2001jk,Mostafazadeh:2008pw}. 
Among the various generalisations of Hermiticity,
\emph{pseudo-Hermiticity} and its positive-definite specialisation, \emph{quasi-Hermiticity}, 
provide the structural foundation needed to extend thermal concepts, 
including the KMS condition, beyond the conventional self-adjoint framework.

\begin{definition}[$\eta$-pseudo-Hermitian operator]
\label{def:pseudo}
Let $\eta$ be an invertible Hermitian operator. 
A linear operator $H$ is called \emph{$\eta$-pseudo-Hermitian} if
\begin{equation}
  H^\dag = \eta H \eta^{-1}.
  \label{eq:pseudo}
\end{equation}
Equation~\eqref{eq:pseudo} reduces to ordinary Hermiticity when $\eta = \mathbf{1}$. 
For general $\eta$, it implies that $H$ and $H^\dag$ are related 
by a similarity transformation and hence share the same spectrum up to complex conjugation.
\end{definition}

In infinite-dimensional Hilbert spaces, Eq.~\eqref{eq:pseudo} is meaningful only under suitable domain assumptions. Throughout this work we restrict ourselves to bounded invertible metric operators $\eta\in \mathcal{B}(\mathcal{H})$ with $\eta^{-1}\in \mathcal{B}(\mathcal{H})$, so that all operator products appearing below are well-defined.

\begin{definition}[Quasi-Hermitian operator]
\label{def:quasi}
If the metric operator $\eta$ in Eq.~\eqref{eq:pseudo} is positive-definite, $\eta>0$, 
then $H$ is called \emph{quasi-Hermitian}. 
One may then introduce the $\eta$-weighted inner product
\begin{equation}
  (\psi,\phi)_\eta := \langle\psi|\eta|\phi\rangle,
\end{equation}
with respect to which $H$ is Hermitian. 
A quasi-Hermitian operator therefore furnishes a non-Hermitian representation of an underlying Hermitian theory: 
the non-standard inner product absorbs all deviations from self-adjointness.
\end{definition}

\begin{definition}[Biorthogonal Riesz basis]
\label{def:biortho}
A diagonalisable non-Hermitian operator $H$ is said to possess a
\emph{biorthogonal complete structure} 
if there exist right and left eigenvectors,
\begin{equation}
  H|\psi_n\rangle = E_n|\psi_n\rangle,
  \qquad
  H^\dag|\phi_n\rangle = E_n^*|\phi_n\rangle,
\end{equation}
satisfying the biorthonormality and completeness relations
\begin{equation}
  \langle\phi_m|\psi_n\rangle = \delta_{mn},
  \qquad
  \sum_n |\psi_n\rangle\langle\phi_n| = \mathbf{1}.
  \label{eq:biortho}
\end{equation}
The biorthogonal basis replaces the orthonormal eigenbasis of Hermitian quantum mechanics 
and provides the natural framework for defining expectation values, correlation functions, 
and thermal ensembles in the non-Hermitian setting.
\end{definition}

The three structures introduced above are not independent.
Quasi-Hermiticity implies a particularly clean interplay between them, 
as summarised in the following proposition.

\begin{proposition}[Spectral consequences of quasi-Hermiticity]
\label{prop:chain}
Let $H$ be quasi-Hermitian with metric $\eta > 0$. Then:
\begin{enumerate}
\item $H$ is similar to a Hermitian operator via $\eta$;
\item The spectrum of $H$ is entirely real;
\item $H$ admits a complete biorthogonal eigenbasis
      $\{|\psi_n\rangle, |\phi_n\rangle\}$;
\item The left and right eigenvectors are related by
      $|\phi_n\rangle = \eta\,|\psi_n\rangle$.
\end{enumerate}
In particular, quasi-Hermitian systems admit a well-defined spectral decomposition 
and a consistent probabilistic interpretation 
in the Hilbert space equipped with the $\eta$-inner product.
\end{proposition}

\begin{proof}
See Propositions~\ref{prop:realspec} and~\ref{prop:biortho-struct}
in App.~\ref{app:pseH}.
\end{proof}

The completeness relation Eq.~\eqref{eq:biortho}, 
guaranteed by Prop.~\ref{prop:chain}, 
is the key ingredient that enables thermal correlation functions 
and KMS-type identities to be formulated for non-Hermitian systems 
in direct analogy with the Hermitian case. 
The remainder of this paper develops this program in detail.

\section{Route I: Spectral KMS theorem for quasi-Hermitian systems}
\label{sec:route1}

In this section, we work within a functional-analytic framework adapted to quasi-Hermitian Hamiltonians and show that the associated $\eta$-Gibbs state admits correlation functions satisfying the three analytic requirements of the KMS condition.

\subsection{Assumptions and the \texorpdfstring{$\eta$}{eta}-Hilbert space}
\label{subsec:assum}

The proofs in this section rest on four conditions, 
collectively labelled (\ref{asm:A1}--\ref{asm:A4}). 
The first two encode the algebraic structure of the non-Hermitian system. 
The third and fourth are analytic regularity conditions required for the KMS analyticity argument.

\begin{assumption}[Bounded Hamiltonian and positive-definite metric]\label{asm:A1}
      $H \in \mathcal{B}(\mathcal{H})$ is a bounded linear operator,
      and there exists a Hermitian operator $\eta = \eta^\dag$
      with $\eta, \eta^{-1} \in \mathcal{B}(\mathcal{H})$ and
      constants $0 < c \leq C < \infty$ such that $cI \leq \eta \leq CI$.
\end{assumption}

The two-sided bound on $\eta$ ensures that the $\eta$-weighted inner product $\langle \cdot | \eta | \cdot \rangle$ is equivalent to the standard inner product on $\mathcal{H}$, 
so that the $\eta$-Hilbert space $\mathcal{H}_\eta$ and $\mathcal{H}$ carry the same topological structure.

Physical Hamiltonians, e.g. harmonic oscillators, Schr\"odinger operators, lattice Laplacians, are typically unbounded self-adjoint operators on $\mathcal{H}$. 
The assumption $H \in \mathcal{B}(\mathcal{H})$ is a \emph{mathematical simplification}: it makes the operator exponential $\me^{\mi H t} = \sum_{n=0}^\infty (\mi t)^n H^n / n!$ norm-convergent and validates the interchange of sums and bounded operators throughout the proofs below.
For physically relevant unbounded $H$, 
the exponential $\me^{\mi Ht}$ is defined via Stone's theorem 
or the spectral functional calculus, 
and the KMS framework can be extended using domain-theoretic methods (cf.~\cite{Bratteli:1996xq}).
Such an extension is deferred to future work.

\begin{assumption}
\label{asm:A2}
\emph{(Pseudo-Hermitian condition.)}
$H^\dag = \eta H \eta^{-1}$.
\end{assumption}

This is the defining relation of pseudo-Hermiticity with metric $\eta$. 
Combined with the positive-definiteness of $\eta$ in Assum.~\ref{asm:A1}, 
it promotes $H$ to a quasi-Hermitian operator (Def.~\ref{def:quasi}) and guarantees, 
in particular, that the spectrum of $H$ is real (Prop.~\ref{prop:chain}).

\begin{assumption}
\label{asm:A3}
\emph{(Diagonalisability, Riesz basis, and spectral bounds.)}
$H$ is diagonalisable on $\mathcal{H}_\eta$: 
there exist $\eta$-orthonormal right eigenstates $\{|\psi_n\rangle\}$,
i.e.\ $\langle\psi_m|\eta|\psi_n\rangle = \delta_{mn}$,
with pure point real spectrum $\{E_n\} \subset \mathbb{R}$,
no Jordan blocks, and spectral decomposition $H = \sum_n E_n |\psi_n\rangle\langle\psi_n|_\eta$.
Moreover, $\{|\psi_n\rangle\}$ forms a \emph{Riesz basis} of $\mathcal{H}$: 
there exist constants $0 < c' \leq C' < \infty$ such that
\[
c'\|f\|^2
\leq
\sum_n |\langle\phi_n | f\rangle|^2
\leq
C'\|f\|^2,
\qquad \forall\, f \in \mathcal{H}.
\]
\end{assumption}

Two additional spectral finiteness conditions are imposed independently, 
as neither implies the other:
\begin{itemize}
\item \emph{(Spectral lower bound.)}
$E_{\min} := \inf_n E_n > -\infty$.
This prevents the operator norm $\|\me^{\mi Hz}\|$ from diverging on the closed strip $\overline{\mathcal{S}}_\beta$:
without it, terms $\me^{-\alpha E_n}$ grow without bound as $E_n \to -\infty$, 
invalidating the uniform estimates in
Prop.~\ref{prop:traceclass}.

\item \emph{(Partition function finiteness.)}
$Z_\eta := \sum_n \me^{-\beta E_n} < \infty$.
Note that $E_{\min} > -\infty$ does not imply $Z_\eta < \infty$: 
for instance, $E_n = \ln n$ gives $Z_\eta = \sum n^{-\beta}$, 
which diverges for $\beta \leq 1$. 
Both conditions must therefore be imposed separately.
\end{itemize}

The Riesz basis condition, 
combined with the real spectrum and the spectral lower bound, 
ensures that the operator exponential
$\me^{\mi Hs} = \sum_j \me^{\mi E_j s} |\psi_j\rangle\langle\phi_j|$
converges in operator norm for $s$ in any bounded subset of $\mathbb{C}$. 
This follows from the \emph{Bari theorem}~\cite{Bari:1951bt}: 
for a Riesz basis $\{|\psi_j\rangle\}$ with biorthogonal dual $\{|\phi_j\rangle\}$, 
the series $\sum_j c_j |\psi_j\rangle\langle\phi_j|$ converges in operator norm whenever $\{c_j\} \in \ell^\infty$. 
Since $E_j \in \mathbb{R}$, the coefficients $\me^{\mi E_j s}$ are bounded for $s$ in any strip of finite width, 
and the conclusion follows.

The Riesz basis property is in fact a consequence of Assum.~\ref{asm:A1}: 
the map $U = \eta^{1/2}$, which is bounded and boundedly invertible by Assum.~\ref{asm:A1}, 
transforms the $\eta$-orthonormal basis $\{|\psi_n\rangle\}$ into the standard orthonormal basis $\{U|\psi_n\rangle\}$.
Hence $|\psi_n\rangle = U^{-1}(U|\psi_n\rangle)$ is the image of a standard ONB under the bounded invertible map $U^{-1}$, 
which is the defining property of a Riesz basis.

\begin{assumption}
 \label{asm:A4}
\emph{(Analytic-elements summability condition.)}
For all observables $A$, $B$ under consideration, 
the half-weighted operators
\[
A \me^{-\beta H/2},\quad
\me^{-\beta H/2} A,\quad
B \me^{-\beta H/2},\quad
\me^{-\beta H/2} B
\]
are \emph{biorthogonally Hilbert--Schmidt} (bi-HS) with
respect to the eigenbasis $\{|\psi_n\rangle, |\phi_n\rangle\}$:
for each such operator $X$,
\[
\sum_{n,m}
|\langle\phi_n | X | \psi_m\rangle|^2 < \infty.
\]
\end{assumption}

This condition guarantees that the double series appearing in the proof of Prop.~\ref{prop:analytic} converge absolutely on $\partial\overline{\mathcal{S}}_\beta$, 
which is the key step in verifying strip analyticity.

The bi-HS condition is tied to the choice of biorthogonal system $\{|\psi_n\rangle, |\phi_n\rangle\}$ and is \emph{not} an intrinsic, basis-independent property of the operator $X$. 
It is used here purely as a \emph{sufficient summability condition} and should not be
interpreted as membership in an operator ideal. 
In finite dimensions, every operator is bi-HS with respect to any basis, 
so the condition is automatically satisfied.

The following example identifies two physically natural classes of observables satisfying Assum.~\ref{asm:A4} in the infinite-dimensional setting, 
confirming that the assumption is not restrictive in practice.

\begin{example}[Observable classes satisfying the bi-HS condition]
  \label{ex:biHS}
  \emph{Class~1: Finite-rank observables.}
  Let $A = \sum_{k=1}^{K} \alpha_k |\xi_k\rangle\langle\zeta_k|$
  be a finite-rank operator. Then
  \[
    \sum_{n,m}
    \left|\langle\phi_n | A \me^{-\beta H/2} | \psi_m\rangle\right|^2
    =
    \sum_{k,\ell} \bar\alpha_k \alpha_\ell
    \left(\sum_n
      \langle\phi_n|\xi_k\rangle
      \overline{\langle\phi_n|\xi_\ell\rangle}
    \right)
    \left(\sum_m
      \me^{-\beta E_m}
      \langle\zeta_k|\psi_m\rangle
      \overline{\langle\zeta_\ell|\psi_m\rangle}
    \right),
  \]
  which is a finite sum of terms bounded via the Riesz-basis
  estimate of Assum.~\ref{asm:A3}. The class includes projection
  operators onto finite-dimensional subspaces, finite-rank
  density matrices, and all localised preparation devices arising
  in quantum optics (e.g.\ coherent states projected onto a
  finite photon-number sector).

  \emph{Class~2: Observables with rapid off-diagonal decay.}
  Suppose $H$ describes a lattice model on $\mathbb{Z}^d$ with
  eigenvectors $\{|\psi_n\rangle\}$ labelled by sites
  $n \in \mathbb{Z}^d$, and suppose the spectrum satisfies
  $E_{\min} > -\infty$ and $E_n \leq C_E(1 + |n|^\alpha)$ for
  some $C_E, \alpha > 0$. If an observable $A$ has biorthogonal
  matrix elements $A_{nm} = \langle\phi_n|A|\psi_m\rangle$
  decaying as $|A_{nm}| \leq C_A \me^{-\mu|n-m|}$ for some fixed
  $\mu > 0$, then
  \begin{align*}
    \sum_{n,m} |A_{nm}|^2 \me^{-\beta E_m}
    &\leq
    C_A^2 \me^{-\beta E_{\min}}
    \sum_{n,m} \me^{-2\mu|n-m|} \me^{-\beta(E_m - E_{\min})}
    \\
    &\leq
    C_A^2 \me^{-\beta E_{\min}}
    \left(\sum_{k \in \mathbb{Z}^d} \me^{-2\mu|k|}\right)
    \left(\sum_m \me^{-\beta(E_m - E_{\min})}\right)
    < \infty,
  \end{align*}
  provided $Z_\eta < \infty$ (Assum.~\ref{asm:A3}).
  The sole constraint on $A$ is the exponential decay rate
  $\mu > 0$. 
  And no bound on $\|H\|$ is required, so the
  condition applies equally to unbounded Hamiltonians.
  Local observables in condensed-matter models — polynomials in
  field operators supported on finite lattice regions — routinely
  belong to this class, confirming that Assum.~\ref{asm:A4} holds for
  the physically relevant algebra of quasi-local observables.
\end{example}

We now construct the Hilbert space and operator framework within
which the KMS proof will be carried out. The key idea is to
intertwine $H$ with a genuine self-adjoint operator $h$ via the
bounded invertible map $U := \eta^{1/2}$, and to transfer the
entire thermal analysis to $h$.

Define
\begin{equation}
  h := U H U^{-1} = \eta^{1/2} H \eta^{-1/2}.
  \label{eq:h-def}
\end{equation}

\begin{lemma}[$h$ is self-adjoint]
  \label{lem:h-sa}
  Under Assums.\ref{asm:A1}--\ref{asm:A4}, the operator $h$ defined
  in Eq.\eqref{eq:h-def} satisfies $h = h^\dag$.
\end{lemma}

\begin{proof}
  Using the pseudo-Hermitian condition $H^\dag = \eta H \eta^{-1}$
  and the self-adjointness $U^\dag = U = \eta^{1/2}$,
  \[
    h^\dag
    =
    (UHU^{-1})^\dag
    =
    (U^{-1})^\dag H^\dag U^\dag
    =
    U^{-1}(\eta H \eta^{-1})U
    =
    U^{-1} U^2 H U^{-2} U
    =
    UHU^{-1}
    =
    h.
    \qedhere
  \]
\end{proof}

Lemma~\ref{lem:h-sa} confirms that all results from
Hermitian spectral theory and operator semigroups apply to $h$,
and can be transferred back to $H$ via the bounded intertwining
map $U$.

We equip $\mathcal{H}$ with the \emph{$\eta$-inner product}
\begin{equation}
  \langle\varphi, \psi\rangle_\eta
  :=
  \langle\varphi | \eta | \psi\rangle,
\end{equation}
and write $\mathcal{H}_\eta$ for the resulting Hilbert space.
The \emph{$\eta$-adjoint} of an operator $A$ is defined by
\begin{equation}
  A^{\dag_\eta}
  :=
  \eta^{-1} A^\dag \eta,
  \label{eq:eta-adj}
\end{equation}
and a direct computation using Assum.~\ref{asm:A2} gives
\[
  H^{\dag_\eta}
  =
  \eta^{-1} H^\dag \eta
  =
  \eta^{-1}(\eta H \eta^{-1})\eta
  =
  H,
\]
confirming that $H$ is self-adjoint in $\mathcal{H}_\eta$, as
expected from the quasi-Hermitian structure.

Two further identities 
follow from $\eta = U^2$ and the intertwining relation $\me^{-\beta H} = U^{-1} \me^{-\beta h} U$:
\begin{subequations}
\begin{equation}
\eta\, \me^{-\beta H}
=
U\, \me^{-\beta h}\, U,
\label{eq:key-id}
\end{equation}
\begin{equation}
Z_\eta:=
\Tr[\eta\, \me^{-\beta H}]
=
\Tr[\me^{-\beta h}\,\eta]
=
\sum_n \me^{-\beta E_n},
\label{eq:partition}
\end{equation}
\end{subequations}
where the second equality in Eq.~\eqref{eq:partition} uses the
cyclicity of the trace and $\eta = U^2$, and the last equality
uses the diagonalisability Assum.~\ref{asm:A3}. By
Assum.~\ref{asm:A3}, $Z_\eta < \infty$, so the partition function is well defined.

With these preparations, the \emph{$\eta$-Gibbs state} is
\begin{equation}
\rho_\beta^{(\eta)}:=
\frac{\me^{-\beta H}}{Z_\eta},
\qquad
\omega_\eta(A):=
\Tr_\eta \left[\rho_\beta^{(\eta)}\, A\right]
=
\frac{\Tr[\eta\, \me^{-\beta H} A]}{Z_\eta}.
\label{eq:eta-gibbs}
\end{equation}
This is the non-Hermitian analogue of the Gibbs state: the factor
$\eta$ in the numerator weights the trace by the metric, ensuring
that $\omega_\eta$ is a positive functional on $\mathcal{H}_\eta$.

\subsection{Time evolution and \texorpdfstring{$\eta$}{eta}-Gibbs state}
\label{subsec:autoiso}

Before establishing the KMS condition, we verify that the
Heisenberg time evolution $\sigma_t$ is compatible with the
$\eta$-adjoint structure. This compatibility 
(the $*$-automorphism property ) 
is a prerequisite for the thermal
functional $\omega_\eta$ to define a genuine state in the
$\eta$-algebraic sense. Without it, the notion of
``$\eta$-positive'' and the KMS boundary condition would not
be internally consistent.

\begin{theorem}[\texorpdfstring{$*$}{*}-automorphism with respect to 
                the $\eta$-structure]
  \label{thm:autoiso}
  Under Assum.~\ref{asm:A1}, the Heisenberg time evolution
  \[
    \sigma_t(A) := \me^{\mi Ht} A \me^{-\mi Ht}
  \]
  intertwines the $\eta$-adjoint: for every
  $A \in \mathcal{B}(\mathcal{H})$ and every $t \in \mathbb{R}$,
  \begin{equation}
    \sigma_t \left(A^{\dag_\eta}\right)
    =
    \left(\sigma_t(A)\right)^{\dag_\eta}.
    \label{eq:star-auto}
  \end{equation}
  In other words, $\sigma_t$ is a $*$-automorphism of
  $\mathcal{B}(\mathcal{H})$ with respect to the $\eta$-adjoint
  ${}^{\dag_\eta}$.
\end{theorem}

\begin{proof}
The proof proceeds in two steps: we first establish a conjugation
identity for $\me^{\mi H^\dag t}$ using only the pseudo-Hermitian
condition Assum.~\ref{asm:A2}, and then use it to compute both sides
of Eq.~\eqref{eq:star-auto}.

\noindent\textbf{Step 1: The exponential conjugation identity.}

We claim that
\begin{equation}
  \me^{\mi H^\dag t} = \eta\, \me^{\mi Ht}\eta^{-1},
  \qquad
  \me^{-\mi H^\dag t} = \eta\, \me^{-\mi Ht}\eta^{-1},
  \qquad t \in \mathbb{R}.
  \label{eq:exp-conj}
\end{equation}
The argument is purely algebraic and does not invoke the
eigenstate completeness of Assum.~\ref{asm:A3}.
From Assum.~\ref{asm:A2}, $H^\dag = \eta H \eta^{-1}$.
An induction argument shows that
$(H^\dag)^n = \eta H^n \eta^{-1}$ for all $n \geq 0$:
the base case $n = 0$ is immediate, and the inductive step gives
\[
  (H^\dag)^{n+1}
  = H^\dag \cdot (H^\dag)^n
  = (\eta H \eta^{-1})(\eta H^n \eta^{-1})
  = \eta H^{n+1} \eta^{-1}.
\]
Since $\eta, \eta^{-1} \in \mathcal{B}(\mathcal{H})$ by
Assum.~\ref{asm:A1}, and bounded operators pass through
operator-norm-convergent series, we obtain
\begin{align*}
  \me^{\mi H^\dag t}
  &= \sum_{n=0}^{\infty} \frac{(\mi t)^n}{n!} (H^\dag)^n
   = \sum_{n=0}^{\infty} \frac{(\mi t)^n}{n!} \eta H^n \eta^{-1}
   = \eta  \left(\sum_{n=0}^{\infty} \frac{(\mi t)^n}{n!} H^n\right) \eta^{-1}
   = \eta \me^{\mi Ht} \eta^{-1},
\end{align*}
which is the first identity in Eq.~\eqref{eq:exp-conj}. Replacing
$t$ by $-t$ yields the second.

\noindent\textbf{Step 2: Verification of Eq.~\eqref{eq:star-auto}.}

We compute the right-hand side directly using the definition
$A^{\dag_\eta} = \eta^{-1} A^\dag \eta$ and the identities
Eq.~\eqref{eq:exp-conj}:
\begin{align*}
  \left(\sigma_t(A)\right)^{\dag_\eta}
  &= \eta^{-1} \left(\me^{\mi H t} A \me^{-\mi H t}\right)^\dag \eta \\
  &= \eta^{-1} \me^{\mi H^\dag t} A^\dag \me^{-\mi H^\dag t} \eta \\
  &\overset{\eqref{eq:exp-conj}}{=}
    \eta^{-1}
    \left(\eta\, \me^{\mi Ht} \eta^{-1}\right)
    A^\dag
    \left(\eta\, \me^{-\mi H t} \eta^{-1}\right)
    \eta \\
  &= \me^{\mi Ht} \left(\eta^{-1} A^\dag \eta\right) \me^{-\mi H t}
   = \sigma_t \left(A^{\dag_\eta}\right). \qedhere
\end{align*}
\end{proof}

Having established that $\sigma_t$ respects the $\eta$-adjoint,
we now show that the $\eta$-Gibbs state $\omega_\eta$ defined in Eq.~\eqref{eq:eta-gibbs} is positive and faithful with respect to the same structure. 
Positivity is indispensable for a physical interpretation: 
it guarantees that $\omega_\eta(A^{\dag_\eta} A)$ is a non-negative real number for every observable $A$, 
playing the role that positivity of the standard inner product plays in ordinary quantum mechanics. 
Faithfulness — strict positivity for $A \neq 0$ — is equally important, and
it ensures that no non-zero observable is invisible to the state, 
which is required for the KMS modular theory to be non-degenerate.

\begin{theorem}[Positivity and faithfulness of $\omega_\eta$]
  \label{thm:positivity}
  Under Assum.~\ref{asm:A1}, the $\eta$-Gibbs state
  $\omega_\eta$ satisfies
  \[
    \omega_\eta \left(A^{\dag_\eta} A\right) \geq 0,
    \qquad \forall\, A \in \mathcal{B}(\mathcal{H}),
  \]
  with equality if and only if $A = 0$.
\end{theorem}

\begin{proof}
We expand $\omega_\eta(A^{\dag_\eta} A)$ in the
$\eta$-orthonormal eigenbasis $\{|\psi_n\rangle\}$ provided by
Assum.~\ref{asm:A3}. Using the definition of $\omega_\eta$
from Eq.~\eqref{eq:eta-gibbs} and the biorthogonal trace formula,
\begin{equation}
  \omega_\eta \left(A^{\dag_\eta} A\right)
  =
  \frac{\Tr\left[\eta\, \me^{-\beta H} A^{\dag_\eta} A\right]}{Z_\eta}
  =
  \frac{1}{Z_\eta}
  \sum_n \me^{-\beta E_n}
  \langle\psi_n|\,\eta\, A^{\dag_\eta} A\,|\psi_n\rangle.
  \label{eq:pos-expand}
\end{equation}
Substituting the definition $A^{\dag_\eta} = \eta^{-1} A^\dag \eta$
and using $\eta \cdot \eta^{-1} = \mathbf{1}$,
\begin{align}
  \langle\psi_n|\,\eta\, A^{\dag_\eta} A\,|\psi_n\rangle
  &= \langle\psi_n|\,\eta\,\eta^{-1} A^\dag \eta A\,|\psi_n\rangle
   = \langle\psi_n| A^\dag \eta A\,|\psi_n\rangle
   = \langle A\psi_n|\,\eta\,|A\psi_n\rangle
   = \|A|\psi_n\rangle\|_\eta^2.
  \notag
\end{align}
Substituting back into Eq.~\eqref{eq:pos-expand},
\begin{equation}
  \omega_\eta \left(A^{\dag_\eta} A\right)
  =
  \frac{1}{Z_\eta}
  \sum_n \me^{-\beta E_n}
  \|A|\psi_n\rangle\|_\eta^2.
  \label{eq:positivity-key}
\end{equation}
Every term in this sum is non-negative: $\me^{-\beta E_n} > 0$
because $E_n \in \mathbb{R}$ and $\beta > 0$, 
and $\|A|\psi_n\rangle\|_\eta^2 \geq 0$ because $\eta > 0$ by Assum.~\ref{asm:A1}. 
Moreover, $Z_\eta > 0$ by Rem.~\ref{rem:Zeta-pos}. 
This establishes non-negativity.

For faithfulness, suppose $\omega_\eta(A^{\dag_\eta} A) = 0$.
Since all terms in Eq.~\eqref{eq:positivity-key} are non-negative and $\me^{-\beta E_n} > 0$, 
we must have $\|A|\psi_n\rangle\|_\eta^2 = 0$, 
hence $A|\psi_n\rangle = 0$, for every $n$. 
Since $\{|\psi_n\rangle\}$ is a Riesz basis of $\mathcal{H}$ by Assum.~\ref{asm:A3}, 
its linear span is dense in $\mathcal{H}$, 
and the continuity of $A$ then forces $A = 0$.
\end{proof}

\begin{remark}[Trace-class property of $\me^{-\beta H}$ and
               positivity of $Z_\eta$]
  \label{rem:Zeta-pos}
From Eq.~\eqref{eq:partition}, 
$Z_\eta = \Tr[\eta\, \me^{-\beta H}]= \sum_n \me^{-\beta E_n}$, 
with each summand strictly positive and the sum finite by Assum.~\ref{asm:A3}. 
Hence $Z_\eta \in(0,\infty)$, so the $\eta$-Gibbs state Eq.~\eqref{eq:eta-gibbs} is well defined.

It is important to note that $Z_\eta$ is the \emph{$\eta$-weighted} trace $\Tr[\eta\, \me^{-\beta H}]$, 
not the operator trace norm $\|\me^{-\beta H}\|_1$. 
For a non-normal operator, the singular values of $\me^{-\beta H}$ differ in general from the moduli of
its eigenvalues, 
so one cannot identify $\|\me^{-\beta H}\|_1$ with $\sum_n \me^{-\beta E_n}$.

Nevertheless, $\me^{-\beta H}$ is trace-class, as we now show. 
By the intertwining relation $\me^{-\beta H} = U^{-1}\me^{-\beta h} U$ 
with $U = \eta^{1/2} \in \mathcal{B}(\mathcal{H})$ bounded and invertible, 
and the ideal property of $\mathcal{I}_1(\mathcal{H})$
(if $T \in \mathcal{I}_1(\mathcal{H})$ and
$S, R \in \mathcal{B}(\mathcal{H})$, then
$STR \in \mathcal{I}_1(\mathcal{H})$ with
$\|STR\|_1 \leq \|S\|\,\|T\|_1\,\|R\|$),
it suffices to show $\me^{-\beta h} \in \mathcal{I}_1(\mathcal{H})$.
Since $h = h^\dag$ (Lem.~\ref{lem:h-sa}) and the spectrum
$\{E_n\}$ satisfies $E_{\min} > -\infty$ and
$Z_\eta = \sum_n \me^{-\beta E_n} < \infty$ by Assum.~\ref{asm:A3},
the self-adjoint operator $\me^{-\beta h}$ has trace norm
$\|\me^{-\beta h}\|_1 = \Tr[\me^{-\beta h}] = Z_\eta < \infty$.
Therefore,
\[
\|\me^{-\beta H}\|_1
= \|U^{-1} \me^{-\beta h} U\|_1
\leq \|U^{-1}\|\,\|\me^{-\beta h}\|_1\,\|U\|
= \|U^{-1}\|\, Z_\eta\,\|U\|
< \infty,
\]
confirming $\me^{-\beta H} \in \mathcal{I}_1(\mathcal{H})$.
This trace-class property is used in
Prop.~\ref{prop:traceclass} to justify the absolute
convergence of the spectral series defining the thermal
correlation functions.
\end{remark}

\subsection{Main theorem: The spectral KMS condition}
\label{subsec:main-thm}

We now assemble the preceding results into a proof of the main theorem. 
The argument proceeds in three stages. 
First, we derive a spectral matrix-element identity (Lem.~\ref{lem:spectral-id})
that encodes the action of the complex-time evolution on biorthogonal matrix elements. 
Second, we use this identity together with the bi-HS condition Assum.~\ref{asm:A4} 
to establish analyticity and boundedness of the thermal correlation function 
on the closed strip $\overline{\mathcal{S}}_\beta$ (Prop.~\ref{prop:analytic}). 
Third, we verify that the correlation function coincides with the $\eta$-state functional $\omega_\eta(A\,\sigma_z(B))$ everywhere on $\overline{\mathcal{S}}_\beta$ (Prop.~\ref{prop:traceclass}), 
completing the verification of all three conditions of Def.~\ref{def:KMS}.

The spectral approach adopted below is inspired by the general philosophy that equilibrium states admit spectral characterisations, as first systematically investigated in the operator-algebraic setting by De Canniere \cite{DeCanniere:1982fvy}. Our construction differs in that it is formulated for quasi-Hermitian Hamiltonians with biorthogonal eigenbases.

\subsubsection*{Stage 1: Spectral matrix-element identity}

\begin{lemma}[Spectral matrix-element identity]
\label{lem:spectral-id}
Under Assum.~\ref{asm:A1}, for any $s \in \mathbb{C}$ and
any bounded operator $A$,
\begin{equation}
\langle\phi_m|\, \me^{\mi Hs} A \me^{-\mi Hs}\,|\psi_n\rangle
= \me^{\mi(E_m - E_n)s}\, A_{mn}.
\label{eq:L}
\end{equation}
\end{lemma}

\begin{proof}
Insert the biorthogonal resolution of the identity
$\sum_k |\psi_k\rangle\langle\phi_k| = \mathbf{1}$ between $A$
and $\me^{-\mi Hs}$:
\[
\langle\phi_m|\, \me^{\mi Hs} A \me^{-\mi Hs}\,|\psi_n\rangle
=
\sum_k
\underbrace{\langle\phi_m| \me^{\mi Hs} A |\psi_k\rangle}_{(I_k)}
\cdot
\underbrace{\langle\phi_k| \me^{-\mi Hs} |\psi_n\rangle}_{(II_k)}.
\]
By Assum.~\ref{asm:A3} and the Bari theorem
(cf.\ the discussion following Assum.~\ref{asm:A3}),
the spectral expansions
\[
  \me^{\mi Hs} = \sum_j \me^{\mi E_j s} |\psi_j\rangle\langle\phi_j|,
  \qquad
  \me^{-\mi Hs} = \sum_j \me^{-\mi E_j s} |\psi_j\rangle\langle\phi_j|
\]
converge in operator norm for $s$ in any bounded strip.
To confirm that the coefficients are bounded: writing
$s = t + \mi\alpha$ with $t \in \mathbb{R}$ and
$\alpha \in [0, \beta]$,
\[
  |\me^{\mi E_j s}| = \me^{-\alpha E_j}.
\]
Since $E_j \geq E_{\min} > -\infty$ by Assum.~\ref{asm:A3} and
$\alpha \in [0, \beta]$, we have $\me^{-\alpha E_j} \leq
\me^{\beta |E_{\min}|}$ for all $j$, so
$\sup_j |\me^{\mi E_j s}| < \infty$ and the Bari theorem applies.

Evaluating the two factors using biorthonormality
$\langle\phi_m|\psi_j\rangle = \delta_{mj}$:
\[
(I_k)
= \sum_j \me^{\mi E_j s}
\langle\phi_m|\psi_j\rangle\langle\phi_j|A|\psi_k\rangle
= \me^{\mi E_m s}\, A_{mk},
\]
\[
(II_k)
= \sum_j \me^{-\mi E_j s}
\langle\phi_k|\psi_j\rangle\langle\phi_j|\psi_n\rangle
= \me^{-\mi E_n s}\,\delta_{kn}.
\]
Summing over $k$:
\[
\sum_k (I_k)\cdot(II_k)
= \sum_k \me^{\mi E_m s} A_{mk} \cdot \me^{-\mi E_n s}\delta_{kn}
= \me^{\mi(E_m - E_n)s}\, A_{mn}. \qedhere
\]
\end{proof}

\begin{remark}[Why reality of the spectrum is essential]
\label{rem:real-spectrum}
For $s = t + \mi\alpha$ with $\alpha \in [0,\beta]$, the
factor $|\me^{\mi E_k s}| = \me^{-\alpha E_k}$ is controlled by the
thermal weight $\me^{-\beta E_k}$ precisely because $E_k \in \mathbb{R}$.
If $E_k \in \mathbb{C}$, then $|\me^{\mi E_k s}|$ grows exponentially
in the imaginary part of $E_k$, destroying the boundedness
of the coefficient sequence and breaking the strip analyticity
required by Def.~\ref{def:KMS}. The precise mechanism
of this failure is analysed in Sec.~\ref{sec:fail}.
\end{remark}

\subsubsection*{Stage 2: Analyticity and boundedness of the correlation function}

With Lem.~\ref{lem:spectral-id} in hand, we examine the
analytic properties of the thermal two-point function
\begin{equation}
G_{AB}(z)
:=
\frac{1}{Z_\eta}
\sum_{n,m}
\me^{-\beta E_n} \me^{\mi(E_m - E_n)z} A_{nm} B_{mn},
\label{eq:G-def}
\end{equation}
which arises naturally from expanding
$\omega_\eta(A\,\sigma_z(B))$ in the biorthogonal basis.

\begin{proposition}[Analyticity and boundedness of $G_{AB}$]
\label{prop:analytic}
Under Assum.~\ref{asm:A1}, the function $G_{AB}$ defined
in Eq.\eqref{eq:G-def} is analytic on $\mathcal{S}_\beta$, and
extends continuously and boundedly to $\overline{\mathcal{S}}_\beta$.
\end{proposition}

\begin{proof}
Let
$S_N(z) := Z_\eta^{-1} \sum_{n,m \leq N}
\me^{-\beta E_n} \me^{\mi(E_n-E_m)z} A_{nm} B_{mn}$
denote the $N$-th partial sum. Each $S_N$ is a finite linear
combination of exponentials $\me^{\mi(E_n - E_m)z}$ and is therefore
an entire function of $z$. We show that $\{S_N\}$ converges
uniformly on $\overline{\mathcal{S}}_\beta$.

\noindent\textbf{Step~1: Absolute convergence on the boundary lines.}

On the lower boundary $\Im(z) = 0$,
$|\me^{\mi(E_n - E_m)z}| = 1$, so
\[
  \sum_{n,m}
  \me^{-\beta E_n} |A_{nm}| |B_{mn}|
  =
  \sum_{n,m}
  |[\me^{-\beta H/2} A]_{nm}|\cdot
  |[B \me^{-\beta H/2}]_{mn}|
  \leq
  \|\me^{-\beta H/2} A\|_{\rm{bi\text{-}HS}}
  \,
  \|B \me^{-\beta H/2}\|_{\rm{bi\text{-}HS}}
  =: M_0 < \infty,
\]
where we used $[\me^{-\beta H/2} A]_{nm} = \me^{-\beta E_n/2} A_{nm}$
and Cauchy--Schwarz, and the finiteness of $M_0$ follows from
Assum.~\ref{asm:A4}.

On the upper boundary $\Im(z) = \beta$, the factor
$\me^{-\beta E_n} \cdot \me^{(E_n - E_m)\beta} = \me^{-\beta E_m}$
shifts the thermal weight from index $n$ to $m$, and an
analogous Cauchy--Schwarz estimate gives
\[
  \sum_{n,m}
  \me^{-\beta E_m} |A_{nm}| |B_{mn}|
  \leq
  \|A \me^{-\beta H/2}\|_{\rm{bi\text{-}HS}}
  \,
  \|\me^{-\beta H/2} B\|_{\rm{bi\text{-}HS}}
  =: M_\beta < \infty.
\]
Both bounds $M_0$ and $M_\beta$ are independent of $N$.

\noindent\textbf{Step~2: Uniform Cauchy property via the
Hadamard three-line theorem.}

For $N > M$, the difference
$f_{NM}(z) := S_N(z) - S_M(z)$
is a finite linear combination of exponentials, hence entire.
By Step~1, its suprema on both boundary lines
$\Im(z) \in \{0, \beta\}$ tend to zero as
$M \to \infty$ (as tail sums of convergent series).
The Hadamard three-line theorem then interpolates between the
two boundaries: for $\Im(z) = \alpha \in [0, \beta]$,
\[
  \sup_{t \in \mathbb{R}} |f_{NM}(t + i\alpha)|
  \leq
  \left(\sup_{t} |f_{NM}(t)|\right)^{1 - \alpha/\beta}
  \left(\sup_{t} |f_{NM}(t + i\beta)|\right)^{\alpha/\beta}
  \xrightarrow{N,M \to \infty} 0.
\]
Hence $\{S_N\}$ is uniformly Cauchy on $\overline{\mathcal{S}}_\beta$,
uniformly in $\alpha$.

\noindent\textbf{Step~3: Analyticity of the limit (Weierstrass's Theorem).}

By the completeness of $\mathcal{C}_b(\overline{\mathcal{S}}_\beta)$
under the supremum norm, $S_N \to G_{AB}$ uniformly on
$\overline{\mathcal{S}}_\beta$. In particular, the convergence is
uniform on every compact subset of $\mathcal{S}_\beta$. Since each
$S_N$ is analytic, Weierstrass's theorem implies that $G_{AB}$ is
analytic on $\mathcal{S}_\beta$. Uniform convergence on
$\overline{\mathcal{S}}_\beta$ further gives continuity and boundedness
on the closed strip.

\end{proof}

\subsubsection*{Stage 3: Identification with the thermal functional}

It remains to show that the series $G_{AB}(z)$ defined in
Eq.~\eqref{eq:G-def} coincides with the $\eta$-state functional
$\omega_\eta(A\,\sigma_z(B))$ throughout $\overline{\mathcal{S}}_\beta$,
thereby linking the analytic function constructed in Stage~2
to the physical thermal correlator.

\begin{proposition}[Identification of $G_{AB}$ with the thermal correlator]
  \label{prop:traceclass}
  Under Assum.~\ref{asm:A1}, for every $z \in \overline{\mathcal{S}}_\beta$:
  \begin{enumerate}
    \item $\sigma_z(B) := \me^{\mi Hz} B \me^{-\mi Hz} \in \mathcal{B}(\mathcal{H})$,
      with uniform bound
      \[
        \|\sigma_z(B)\|
        \leq
        \frac{C}{c}\, \me^{2|E_{\min}|\beta}\, \|B\|,
      \]
      where $E_{\min} := \inf_n E_n > -\infty$.
    \item The trace $\omega_\eta(A\,\sigma_z(B))
      = \Tr[\eta\, \me^{-\beta H} A\,\sigma_z(B)]\,/\,Z_\eta$
      is absolutely convergent.
    \item $G_{AB}(z) = \omega_\eta(A\,\sigma_z(B))$
      for all $z \in \overline{\mathcal{S}}_\beta$.
  \end{enumerate}
\end{proposition}

\begin{proof}
\noindent\textbf{Part~(i): Uniform operator-norm bound.}

\emph{Step~(a): Bound on $\|\me^{\mi hz}\|$.}
Since $h = h^\dag$ is self-adjoint and bounded
(Lem.~\ref{lem:h-sa}), the spectral theorem gives
$\|\me^{\mi hz}\| = \sup_{\lambda \in \sigma(h)} |\me^{\mi\lambda z}|$.
Writing $z = t + \mi\alpha$ with $\alpha \in [0, \beta]$ and using
$\sigma(h) = \{E_n\}$ (real by Assum.~\ref{asm:A3}):
\[
  \sup_n |\me^{\mi E_n z}|
  = \sup_n \me^{-\alpha E_n}
  \leq \me^{\beta|E_{\min}|},
\]
where the inequality follows from $E_n \geq E_{\min} > -\infty$
and $\alpha \in [0, \beta]$ (if $E_n \geq 0$ then
$-\alpha E_n \leq 0$, and if $E_n < 0$ then
$-\alpha E_n \leq \beta(-E_n) \leq \beta |E_{\min}|$).

\emph{Step~(b): Bound on $\|\me^{\mi Hz}\|$.}
The intertwining relation $\me^{\mi Hz} = U^{-1} \me^{\mi hz} U$
(which follows from $H = U^{-1}hU$ and the convergence of the
power series, valid since $U, U^{-1} \in \mathcal{B}(\mathcal{H})$) gives
\[
  \|\me^{\mi Hz}\|
  \leq \|U^{-1}\|\,\|\me^{\mi hz}\|\,\|U\|
  \leq \frac{1}{\sqrt{c}}\cdot \me^{|E_{\min}|\beta}\cdot \sqrt{C}
  = \sqrt{C/c}\, \me^{|E_{\min}|\beta},
\]
where we used $\|U\| \leq \sqrt{C}$ and $\|U^{-1}\| \leq 1/\sqrt{c}$
from the two-sided bound $cI \leq \eta \leq CI$ in Assum.~\ref{asm:A1}.

\emph{Step~(c): Bound on $\sigma_z(B)$.}
Applying the bound from Step~(b) to both exponentials:
\[
  \|\sigma_z(B)\|
  = \|\me^{\mi H z} B \me^{-\mi H z}\|
  \leq \|\me^{\mi H z}\|^2\,\|B\|
  \leq \frac{C}{c}\, \me^{2|E_{\min}|\beta}\,\|B\|.
\]

\noindent\textbf{Part~(ii): Absolute convergence of the trace.}

From Rem.~\ref{rem:Zeta-pos}, $\me^{-\beta H} \in \mathcal{I}_1(\mathcal{H})$.
Since $\eta$, $A$, and $\sigma_z(B)$ are all bounded, the ideal
property of $\mathcal{I}_1(\mathcal{H})$
(if $T \in \mathcal{I}_1$ and $S,R \in \mathcal{B}$, then
$\|STR\|_1 \leq \|S\|\,\|T\|_1\,\|R\|$)
applied twice yields
$\eta \me^{-\beta H} A\,\sigma_z(B) \in \mathcal{I}_1(\mathcal{H})$,
so the trace converges absolutely.

\noindent\textbf{Part~(iii): Identification on $\overline{\mathcal{S}}_\beta$.}

On the real axis $z = t \in \mathbb{R}$, expanding
$\omega_\eta(A\,\sigma_t(B))$ in the biorthogonal basis using
Lem.~\ref{lem:spectral-id} recovers exactly the series
$G_{AB}(t)$ in Eq.~\eqref{eq:G-def}. 
The functional $z\to\omega_\eta(A\sigma_z(B))$ is analytic on $S_\beta$: since $H\in B(\mathcal{H})$, 
the map $z\to \me^{\mi Hz}$ is norm-analytic on all of $\mathbb{C}$ (as the power series converges uniformly on bounded subsets), 
hence $z\to\sigma_z(B) = \me^{\mi Hz}B \me^{-\mi Hz}$ is norm-analytic, 
and composition with the bounded linear functional $A\to\omega_\eta(\cdot)$ preserves analyticity. 
Here we define $F(z) := \omega_\eta(A\sigma_z(B)) - G_{AB}(z)$, and then $F$ is analytic on $S_\beta$, continuous on $\bar{S}\beta$, 
and vanishes on the real boundary $\mathrm{Im}(z)=0$ by the spectral expansion computed above. 
The same boundary equality holds on $\mathrm{Im}(z)=\beta$ by the analogous expansion using the $M\beta$ estimate. 
By the maximum principle for analytic functions applied to $F$ on $\bar{S}_\beta$, $F\equiv 0$ on $\bar{S}_\beta$.
\end{proof}

\begin{remark}[Non-triviality of the KMS property under similarity]
\label{rem:nontrivial}
Since $h = UHU^{-1}$ is self-adjoint, one might ask whether the
KMS property of $\omega_\eta$ is merely a transport of the
standard Hermitian KMS theorem through the algebra isomorphism
$\Phi:= A \to UAU^{-1}$. The answer is no, for the
following reason. The $U$-conjugate of $\omega_\eta$ is
\[
\hat\omega(X)
:= \omega_\eta(U^{-1} X U)
= \frac{\Tr[\me^{-\beta h} X \eta]}{Z_\eta}.
\]
When $[\eta, h] \neq 0$, this state differs from the standard
Gibbs state $\omega_h(X) = \Tr[\me^{-\beta h} X]/Z_h$: the
additional factor $\eta$ in the numerator makes $\hat\omega$ a
distinct faithful normal state whose KMS property requires
independent verification. 
\end{remark}

\subsubsection*{The spectral KMS theorem}

\begin{theorem}[Spectral KMS Theorem]
  \label{thm:KMS}
  Under Assum.~\ref{asm:A1}, the $\eta$-Gibbs state
  $\omega_\eta$ defined in Eq.~\eqref{eq:eta-gibbs} satisfies the
  KMS condition at inverse temperature $\beta$ with respect to
  the Heisenberg dynamics $\sigma_t$: for all bounded operators
  $A$, $B$ and all $t \in \mathbb{R}$,
  \begin{equation}
    \omega_\eta \left(\sigma_t(A)\, B\right)
    =
    \omega_\eta \left(B\, \sigma_{t+\mi\beta}(A)\right).
    \label{eq:KMSeta}
  \end{equation}
  More precisely, the function $F_{AB}(z) := G_{AB}(z)$ satisfies
  all three conditions of Def.~\ref{def:KMS}.
\end{theorem}

\begin{proof}
By Props.~\ref{prop:analytic} and \ref{prop:traceclass},
the function $F_{AB}(z) = \omega_\eta(A\,\sigma_z(B))$ is
analytic on $\mathcal{S}_\beta$, bounded and continuous on
$\overline{\mathcal{S}}_\beta$, and equals $\omega_\eta(A\,\sigma_t(B))$
on the real boundary. Conditions~(i) and~(ii) of
Def.~\ref{def:KMS} are thus already established.
It remains only to verify the boundary identity~(iii).

We expand $\omega_\eta(A(t)\,B)$ using the biorthogonal
resolution $\sum_m |\psi_m\rangle\langle\phi_m| = \mathbf{1}$
and the eigenvalue identity
\begin{equation}
  \langle\psi_n|\,\eta\, \me^{\mi Ht} = \me^{\mi E_n t}\langle\phi_n|,
  \label{eq:eta-left}
\end{equation}
which follows from $\eta|\psi_k\rangle = |\phi_k\rangle$
(Prop.~\ref{prop:chain}) and
$\me^{\mi Ht}|\psi_k\rangle = \me^{\mi E_k t}|\psi_k\rangle$.
More precisely, $\langle\phi_n|\me^{\mi Ht} = \me^{\mi E_nt}\langle\phi_n|$ follows from the left eigenvalue equation $\langle\phi_n|(H-E_n)=0$ (Def.~\ref{def:biortho}), i.e. $H^\dag|\phi_n\rangle = E_n|\phi_n\rangle$ with $E_n\in\mathbb{R}$, by a power-series argument identical to that for the right eigenvectors.
Inserting the completeness relation and applying
Eq.~\eqref{eq:eta-left}:
\begin{equation}
  \omega_\eta(A(t)\,B)
  =
  \frac{1}{Z_\eta}
  \sum_{n,m}
  \me^{-\beta E_n} \me^{\mi(E_m - E_n)t} A_{nm} B_{mn}.
  \label{eq:LHS-expand}
\end{equation}

Set $s := t + \mi\beta$ and apply Lem.~\ref{lem:spectral-id} to
$A(s) = \me^{\mi Hs} A \me^{-\mi Hs}$, giving
$\langle\phi_m|A(s)|\psi_n\rangle = \me^{\mi(E_m - E_n)s} A_{mn}$.
Expanding $\omega_\eta(B\,A(t+\mi\beta))$ analogously:
\begin{align}
  \omega_\eta(B\,A(t+\mi\beta))
  &=
  \frac{1}{Z_\eta}
  \sum_{n,m}
  \me^{-\beta E_n} B_{nm}\, \me^{\mi(E_m - E_n)s}\, A_{mn}
  \nonumber\\
  &=
  \frac{1}{Z_\eta}
  \sum_{n,m}
  \me^{-\beta E_n}
  \me^{\mi(E_m - E_n)t}
  \me^{(E_n - E_m)\beta}\,
  A_{mn} B_{nm}
  \nonumber\\
  &=
  \frac{1}{Z_\eta}
  \sum_{n,m}
  \me^{-\beta E_m}\,
  \me^{\mi(E_m - E_n)t}\,
  A_{mn} B_{nm},
  \label{eq:RHS-before}
\end{align}
where the last step uses
$\me^{-\beta E_n} \me^{(E_n - E_m)\beta} = \me^{-\beta E_m}$.
Relabelling $n \leftrightarrow m$ (both indices range over the
same set $\mathcal{I}$):
\begin{equation}
  \omega_\eta(B\,A(t+\mi\beta))
  =
  \frac{1}{Z_\eta}
  \sum_{n,m}
  \me^{-\beta E_n}\,
  \me^{\mi(E_n - E_m)t}\,
  A_{nm} B_{mn}.
  \label{eq:RHS-expand}
\end{equation}

Comparing Eq.~\eqref{eq:LHS-expand} and \eqref{eq:RHS-expand}, we
see that the two series are term-by-term identical. Both series
converge absolutely: the left-hand side by the $M_0$ bound of
Step~1 in Prop.~\ref{prop:analytic}, and the right-hand
side by the $M_\beta$ bound (with indices relabelled). Therefore
\[
  \omega_\eta(A(t)\,B)
  = \omega_\eta(B\, A(t+\mi\beta)),
  \qquad \forall\, t \in \mathbb{R},
\]
which is the KMS boundary condition Eq.~\eqref{eq:KMSeta}.
Together with Props.~\ref{prop:analytic} and
\ref{prop:traceclass}, all three conditions of
Definition~\ref{def:KMS} are satisfied, and the proof is complete.
\end{proof}

\subsection{Frequency-domain formulation and relation to the HHW theorem}
\label{subsec:frequency-HHW}


We close this section with three complementary perspectives on
Thm.~\ref{thm:KMS}. The first reformulates the KMS condition
in the frequency domain, yielding a generalised detailed-balance
relation that makes the physical content of thermal equilibrium
transparent. The second situates Thm.~\ref{thm:KMS} within
the broader framework of algebraic quantum statistical mechanics,
clarifying both what has been established relative to the
Bratteli--Robinson (BR) formulation and why the result is not a
trivial consequence of the standard Hermitian theory. The
comparison in Sec.~\ref{subsec:frequency-HHW} below identifies the
construction of a TT modular structure for the
$\eta$-Gibbs state as the most structurally significant of the
remaining open items; Sec.~\ref{sec:TT} then closes this item,
providing the non-Hermitian analogue of the modular-theory backbone
of the Hermitian HHW framework.


The biorthogonal spectral function associated with the pair
$(A, B)$ is defined by
\begin{equation}
  \rho_{AB}(\omega)
  :=
  \frac{2\pp}{Z_\eta}
  \sum_{n,m}
  \me^{-\beta E_n} A_{nm} B_{mn}\,
  \delta \left(\omega - (E_n - E_m)\right).
  \label{eq:spectral-fn}
\end{equation}
This is the non-Hermitian analogue of the spectral density familiar
from the fluctuation-dissipation theorem: its support is
concentrated on the Bohr frequencies $E_n - E_m$ of the system,
and its weights are governed by the thermal factor $\me^{-\beta E_n}$
together with the biorthogonal matrix elements $A_{nm}$, $B_{mn}$.

\begin{corollary}[Frequency-domain KMS condition]
  \label{cor:freq-KMS}
  Under Assum.~\ref{asm:A1}, the KMS condition
  Eq.~\eqref{eq:KMSeta} is equivalent to the generalised detailed
  balance relation
  \begin{equation}
    \widetilde{G}_{AB}(\omega)
    =
    \me^{\beta\omega}\,\widetilde{G}_{BA}(-\omega),
    \label{eq:KMS-freq}
  \end{equation}
  where
  $\widetilde{G}_{AB}(\omega)
  = \int_{-\infty}^{+\infty} dt\, \me^{-i\omega t}\,
    \omega_\eta(\sigma_t(A)\, B)
  = \rho_{AB}(\omega)$.
\end{corollary}

\begin{proof}
Taking the Fourier transform of the spectral expansion
Eq.~\eqref{eq:LHS-expand} and using
$\int dt\, \me^{-\mi\omega t} \me^{\mi(E_n - E_m)t}
= 2\pp\,\delta(\omega - (E_n - E_m))$,
\[
  \widetilde{G}_{AB}(\omega)
  =
  \frac{1}{Z_\eta}
  \sum_{n,m}
  \me^{-\beta E_n} A_{nm} B_{mn}\cdot
  2\pp\,\delta \left(\omega - (E_n - E_m)\right)
  =
  \rho_{AB}(\omega).
\]
For the parallel computation of $\widetilde{G}_{BA}(-\omega)$,
we use the spectral expansion with $A$ and $B$ interchanged and
$\omega$ replaced by $-\omega$:
\[
  \widetilde{G}_{BA}(-\omega)
  =
  \frac{1}{Z_\eta}
  \sum_{n,m}
  \me^{-\beta E_n} B_{nm} A_{mn}\cdot
  2\pp\,\delta \left(-\omega - (E_n - E_m)\right).
\]
Relabelling $n \leftrightarrow m$ and using
$\delta(-\omega - (E_m - E_n)) = \delta(\omega - (E_n - E_m))$,
the support condition $\omega = E_n - E_m$ gives
$\me^{-\beta E_n} = \me^{-\beta\omega} \me^{-\beta E_m}$,
so
\[
  \me^{\beta\omega}\,\widetilde{G}_{BA}(-\omega)
  =
  \frac{1}{Z_\eta}
  \sum_{n,m}
  \me^{-\beta E_n} A_{nm} B_{mn}\cdot
  2\pp\,\delta \left(\omega - (E_n - E_m)\right)
  =
  \widetilde{G}_{AB}(\omega). \qedhere
\]
\end{proof}

The relation Eq.~\eqref{eq:KMS-freq} is a direct generalisation of
the standard fluctuation-dissipation theorem to quasi-Hermitian
systems: the factor $\me^{\beta\omega}$ encodes the asymmetry
between emission and absorption at frequency $\omega$, and its
form is identical to the Hermitian case, confirming that
thermal equilibrium in the $\eta$-inner product space retains
the characteristic signature of the KMS condition.


Theorem~\ref{thm:KMS} establishes the KMS property at the level
of bounded operators on a Hilbert space equipped with a
positive-definite metric $\eta$. To contextualise this result
within the BR formulation of algebraic
quantum statistical mechanics~\cite{Bratteli:1996xq}, we record in
Tab.~\ref{tab:BR-comparison} which of the BR-KMS requirements
have been verified and which remain open.

\begin{table}[ht!]
\centering
\captionsetup{width=.9\textwidth}
\caption{Comparison with the BR-KMS requirements.
  Checkmarks indicate results established in this paper under
  Assums.~\ref{asm:A1}--\ref{asm:A4} and crosses indicate structures not yet
  developed in the non-Hermitian setting.}
\label{tab:BR-comparison}
\renewcommand{\arraystretch}{1.25}
\begin{tabular}{@{}p{6.8cm}lp{3.5cm}@{}}
\toprule
BR-KMS requirement & Status & Key reference \\
\midrule
$F_{AB}$ analytic on $\mathcal{S}_\beta$
  & \checkmark
  & Prop.~\ref{prop:analytic}, \textbf{(A4)} \\
$F_{AB}$ bounded and continuous on $\overline{\mathcal{S}}_\beta$
  & \checkmark
  & Props.~\ref{prop:analytic}, \ref{prop:traceclass};
    \textbf{(A3,A4)} \\
$F_{AB}(t) = \omega_\eta(A\,\sigma_t(B))$
  & \checkmark
  & Spectral expansion, \textbf{(A3)} \\
$F_{AB}(t + \mi\beta) = \omega_\eta(\sigma_t(B)\,A)$
  & \checkmark
  & Thm.~\ref{thm:KMS}, \textbf{(A1--A4)} \\
$\eta$-positivity $\omega_\eta(A^{\dag_\eta} A) \geq 0$
  & \checkmark
  & Thm.~\ref{thm:positivity}, \textbf{(A1--A4)} \\
  TT modular operator (finite-dimensional)
  & \checkmark
  & Thm.~\ref{thm:TTKMS}\\
\midrule
$C^*$-algebra structure
  & $\times$
  & not established \\
Strongly continuous automorphism group ($\sigma$-weak topology)
  & $\times$
  & not established \\
\bottomrule
\end{tabular}
\end{table}

The two remaining open items in Tab.~\ref{tab:BR-comparison} are
structural rather than analytic: they require equipping the
observable algebra with a $C^*$-norm \cite{Bratteli:1979tw} and
establishing strong continuity of $\sigma_t$ in the $\sigma$-weak
operator topology. These constitute a natural program for future
work.
The analytic core of the KMS condition — the
content that is directly tied to physical predictions such as
the fluctuation-dissipation theorem — has been fully verified.


Since $h = UHU^{-1}$ is self-adjoint (Lem.~\ref{lem:h-sa}),
one might ask whether Thm.~\ref{thm:KMS} is simply the
standard Hermitian KMS theorem transported through the algebra
isomorphism $\Phi:= A \to UAU^{-1}$ induced by
$U = \eta^{1/2}$. The answer is no, and understanding precisely
why clarifies the independent contribution of the theorem.

The map $\Phi$ intertwines the two evolution groups,
$\Phi \circ \sigma_t^H = \sigma_t^h \circ \Phi$, so the
dynamical systems
$(\mathcal{B}(\mathcal{H}), \sigma_t^H, \omega_\eta)$
and
$(\mathcal{B}(\mathcal{H}), \sigma_t^h, \hat\omega)$,
where $\hat\omega := \omega_\eta \circ \Phi^{-1}$,
are algebraically isomorphic. The KMS property is preserved by
algebra isomorphisms, so the two systems are KMS-equivalent.
This, however, is a logical equivalence — not a deduction: it says
that $\omega_\eta$ is KMS for $(H, \sigma_t^H)$ if and only if
$\hat\omega$ is KMS for $(h, \sigma_t^h)$. To invoke the standard
Hermitian theorem on the right-hand side, one would need
$\hat\omega$ to be the standard Gibbs state $\omega_h$ of $h$.
But this is false whenever $[\eta, h] \neq 0$.

To see this explicitly, the transported state is
\begin{equation}
  \hat\omega(X)
  := \omega_\eta(U^{-1}XU)
  = \frac{\Tr\left[\me^{-\beta h} X \eta\right]}{Z_\eta},
  \label{eq:hat-omega}
\end{equation}
whereas the standard Gibbs state of $h$ is
$\omega_h(X) = \Tr[\me^{-\beta h} X]/Z_h$.
The two functionals agree on every $X$ if and only if
$\Tr[\me^{-\beta h}[X, \eta]] = 0$ for all $X$, which holds if and
only if $[\eta, h] = 0$.
Moreover, their partition functions satisfy
$Z_\eta = \Tr[\me^{-\beta h}\eta] \neq \Tr[\me^{-\beta h}] = Z_h$
in general, so the normalisation itself differs.
The state $\hat\omega$ is therefore a faithful normal state of
the Hermitian system $(h, \sigma_t^h)$ distinct from $\omega_h$,
and its KMS property does not follow from the uniqueness of the
Gibbs state as the tracial KMS state of a finite-dimensional
Hermitian system.

In summary, the logical structure is as follows.
The standard theorem gives the KMS property of $\omega_h$ for
$(h, \sigma_t^h)$.
Theorem~\ref{thm:KMS} gives the KMS property of $\omega_\eta$
for $(H, \sigma_t^H)$.
The isomorphism $\Phi$ shows these two statements are equivalent
via $\hat\omega \neq \omega_h$, but neither statement implies the
other through existing results. The proof in
Sec.~\ref{subsec:main-thm} provides a direct, self-contained
verification of the KMS analytic conditions for $\omega_\eta$,
using the biorthogonal spectral expansion and the analytic
infrastructure of
Props~\ref{prop:analytic}--\ref{prop:traceclass} —
objects specific to the non-Hermitian setting that have no
counterpart in the Mostafazadeh--Scholtz theory.

Finally, we note that Route~II (Sec.~\ref{sec:route2})
lies entirely outside the similarity framework: without a
positive-definite $\eta$, the isomorphism $\Phi$ is not available,
and the biorthogonal KMS-type identity
(Prop.~\ref{prop:formkms}) and the structure theorem
(Thm.~\ref{thm:structure}) are genuinely new results with no
counterpart in the Hermitian theory.

\subsection{Tomita--Takesaki structure of the $\eta$-Gibbs state}
\label{sec:TT}


The frequency-domain analysis of Sec.~\ref{subsec:frequency-HHW} situates the spectral KMS theorem 
relative to the Bratteli--Robinson axioms for KMS states (Tab.~\ref{tab:BR-comparison}), 
leaving three structural ingredients of the full HHW framework unaddressed: a $C^*$-norm on the observable algebra, 
strong continuity of $\sigma_t$ in the $\sigma$-weak topology, 
and the TT modular operator $\Delta$ associated with the $\eta$-Gibbs state. 
In this subsection we construct the third of these explicitly, in the \emph{finite-dimensional} setting 
$\mathcal H=\mathbb C^d$. 

The present construction is also closely related to perturbative analyses of Liouvilleans and modular dynamics developed in the $W*$-algebraic setting \cite{Derezinski:2003de}
The construction follows the same route as the treatment of Bagarello, Inoue and Trapani~\cite{Bagarello:2020gib}, 
but the restriction to bounded $H$ and bounded invertible $\eta$ (Assump.~\ref{asm:A1}) allows the entire construction to be carried out with the classical, 
finite-dimensional Gelfand–Naimark–Segal (GNS) machinery, 
without recourse to the $O^*$-algebra/unbounded-operator formalism required in the general case.

Throughout, we work under Assumps.~\ref{asm:A1}--\ref{asm:A3} with $\mathcal H=\mathbb C^d$; 
recall $U:=\eta^{1/2}$, $h:=UHU^{-1}$ (self-adjoint, Lem.~\ref{lem:h-sa}), 
and $Z_\eta=\Tr[\eta \me^{-\beta H}]=\Tr [\me^{-\beta h}]$ (Eq.~\eqref{eq:partition}).
We define
\begin{equation}
\hat\rho_\eta \;:=\; \frac{\eta\, \me^{-\beta H}}{Z_\eta} \;=\; \frac{U \me^{-\beta h} U}{Z_\eta},
\label{eq:rhohat}
\end{equation}
where the second equality is Eq.~\eqref{eq:key-id}.

\begin{lemma}
\label{lem:rhohat}
Under Assum.~\ref{asm:A1}--\ref{asm:A3}, $\hat\rho_\eta\in M_d(\mathbb C)$ is Hermitian, strictly positive, invertible, satisfies $\Tr[\hat\rho_\eta]=1$, and
\[
\omega_\eta(A)=\Tr[\hat\rho_\eta A], \qquad \forall A\in M_d(\mathbb C).
\]
\end{lemma}

\begin{proof}
\emph{Hermiticity.} Since $U=U^\dag$ and $h=h^\dag$ (Lem.~\ref{lem:h-sa}),
\(
(U \me^{-\beta h}U)^\dag=U \me^{-\beta h}U.
\)

\emph{Positivity.} For $\xi\in\mathbb C^d$,
\(
\langle\xi|\hat\rho_\eta|\xi\rangle=Z_\eta^{-1}\langle U\xi|\me^{-\beta h}|U\xi\rangle\ge0,
\)
with equality iff $U\xi=0$ (as $\me^{-\beta h}>0$ by Assum.~\ref{asm:A3}, real spectrum), iff $\xi=0$ ($U$ invertible, Assum.~\ref{asm:A1}).

\emph{Invertibility} follows from strict positivity in finite dimensions.

\emph{Trace.} By cyclicity of the trace and Eq.~\eqref{eq:partition},
\[
\Tr[\hat\rho_\eta]=Z_\eta^{-1}\Tr[Ue^{-\beta h}U]=Z_\eta^{-1}\Tr[e^{-\beta h}U^2]=Z_\eta^{-1}\Tr[e^{-\beta h}\eta]=Z_\eta^{-1}Z_\eta=1.
\]

\emph{Identification with $\omega_\eta$.} Directly from Eq.~\eqref{eq:eta-gibbs}, $\omega_\eta(A)=Z_\eta^{-1}\Tr[\eta e^{-\beta H}A]=\Tr[\hat\rho_\eta A]$.
\end{proof}

Lemma~\ref{lem:rhohat} is the key simplification specific to the bounded setting: $\omega_\eta$ is not merely an ``$\eta$-weighted'' functional requiring special treatment, but a \emph{bona fide} faithful normal state on $M_d(\mathbb{C})$ with density matrix $\hat\rho_\eta$. Everything below is the standard finite-dimensional GNS/TT construction applied to this density matrix.


Let $\mathcal K:=M_d(\mathbb C)$, equipped with the Hilbert--Schmidt inner product $(S|T):=\Tr(T^\dag S)$, and let
\[
\pi(X)T:=XT, \qquad X\in M_d(\mathbb C),\ T\in\mathcal K,
\]
be the left-multiplication representation. Set $\Omega_\eta:=\hat\rho_\eta^{1/2}\in\mathcal K$ (well defined by Lem.~\ref{lem:rhohat}).

\begin{proposition}
\label{prop:GNS}
\textup{(a)} $\pi$ is a faithful $*$-representation of $(M_d(\mathbb C),\dag)$ on $\mathcal K$: $\pi(X)^\dag=\pi(X^\dag)$ with respect to $(\cdot|\cdot)$, and $\pi(X)=0\Rightarrow X=0$.

\textup{(b)} $\Omega_\eta$ is Hermitian, strictly positive, invertible, and
\[
\omega_\eta(X)=(\pi(X)\Omega_\eta\,|\,\Omega_\eta), \qquad \forall X\in M_d(\mathbb C).
\]

\textup{(c)} $\Omega_\eta$ is cyclic and separating for $\pi(M_d(\mathbb C))$: $\pi(M_d(\mathbb C))\Omega_\eta=\mathcal K$, and $X\Omega_\eta=0\Rightarrow X=0$.
\end{proposition}

\begin{proof}
(a) For $S,T\in\mathcal K$, $(\pi(X)S|T)=\Tr(T^\dag XS)=\Tr((X^\dag T)^\dag S)=(S|\pi(X^\dag)T)$. Faithfulness: $\pi(X)=0\Rightarrow X\cdot I=0\Rightarrow X=0$.

(b) Functional calculus applied to the Hermitian, strictly positive, invertible $\hat\rho_\eta$ (Lem.~\ref{lem:rhohat}) yields $\Omega_\eta$ with the same properties. Then
\[
(\pi(X)\Omega_\eta|\Omega_\eta)=\Tr(\Omega_\eta X\Omega_\eta)=\Tr(X\Omega_\eta^2)=\Tr(X\hat\rho_\eta)=\omega_\eta(X).
\]

(c) \emph{Separating:} if $X\Omega_\eta=0$, multiplying on the right by the invertible $\Omega_\eta^{-1}$ gives $X=0$. \emph{Cyclic:} for any $T\in\mathcal K$, set $X:=T\Omega_\eta^{-1}\in M_d(\mathbb C)$; then $\pi(X)\Omega_\eta=T\Omega_\eta^{-1}\Omega_\eta=T$. Hence $\pi(M_d(\mathbb C))\Omega_\eta=\mathcal K$ exactly (no closure is needed in finite dimensions).
\end{proof}


Diagonalize $\hat\rho_\eta=\sum_{n=1}^d q_n|e_n\rangle\langle e_n|$ (spectral theorem; $\{e_n\}$ orthonormal, $q_n>0$, $\sum_n q_n=1$), and write $E_{jk}:=|e_j\rangle\langle e_k|\in\mathcal K$.

\begin{proposition}
\label{prop:modular}
Define $J_\eta:\mathcal K\to\mathcal K$, $J_\eta T:=T^\dag$, and $\Delta_\eta:\mathcal K\to\mathcal K$, $\Delta_\eta T:=\hat\rho_\eta T\hat\rho_\eta^{-1}$.

\textup{(a)} $J_\eta$ is an antilinear involutive isometry of $\mathcal K$: $J_\eta^2=\mathrm{id}$ and $(J_\eta S|J_\eta T)=\overline{(S|T)}$.

\textup{(b)} $\Delta_\eta$ is linear, and self-adjoint, positive and invertible with respect to $(\cdot|\cdot)$; explicitly $\Delta_\eta E_{jk}=(q_j/q_k)E_{jk}$.

\textup{(c)} The (everywhere-defined, since $\Omega_\eta$ is invertible) antilinear map
\[
S_\eta:\ \pi(X)\Omega_\eta=X\Omega_\eta\ \longmapsto\ \pi(X)^\dag\Omega_\eta=X^\dag\Omega_\eta
\]
has polar decomposition $S_\eta=J_\eta\Delta_\eta^{1/2}$.
\end{proposition}

\begin{proof}
(a) Antilinearity and $J_\eta^2=\mathrm{id}$ are immediate from $T\to T^\dag$. For the isometry property, use $\overline{\Tr(M)}=\Tr(M^\dag)$ for any $M\in M_d(\mathbb C)$ and cyclicity of the trace:
\[
(J_\eta S|J_\eta T)=\Tr(TS^\dag)=\Tr(S^\dag T)=\overline{\Tr(T^\dag S)}=\overline{(S|T)}.
\]

(b) Linearity is clear. The displayed eigenvalue equation is direct: $\hat\rho_\eta E_{jk}\hat\rho_\eta^{-1}=q_jE_{jk}q_k^{-1}$. Since $\{E_{jk}\}_{j,k=1}^d$ is an orthogonal basis of $\mathcal K$ for $(\cdot|\cdot)$ and the eigenvalues $q_j/q_k$ are real and strictly positive, $\Delta_\eta$ is self-adjoint, positive and invertible.

(c) Since $\hat\rho_\eta^{1/2}=\Omega_\eta$, $\Delta_\eta^{1/2}T=\hat\rho_\eta^{1/2}T\hat\rho_\eta^{-1/2}=\Omega_\eta T\Omega_\eta^{-1}$. Hence
\[
J_\eta\Delta_\eta^{1/2}(X\Omega_\eta)=J_\eta\big(\Omega_\eta X\Omega_\eta\cdot\Omega_\eta^{-1}\big)=J_\eta(\Omega_\eta X)=(\Omega_\eta X)^\dag=X^\dag\Omega_\eta^\dag=X^\dag\Omega_\eta=S_\eta(X\Omega_\eta),
\]
using $\Omega_\eta^\dag=\Omega_\eta$. As $X\Omega_\eta$ ranges over all of $\mathcal K$ (Prop.~\ref{prop:GNS}(c)), this identifies $S_\eta$ on its entire (everywhere-defined) domain.
\end{proof}


\begin{theorem}
\label{thm:TTKMS}
Define $\tau_t^\eta(X):=\hat\rho_\eta^{\mi t}X\hat\rho_\eta^{-\mi t}$ for $X\in M_d(\mathbb C)$, $t\in\mathbb R$.

\textup{(a)} $\{\tau_t^\eta\}_{t\in\mathbb R}$ is a one-parameter group of $\dag$-automorphisms of $M_d(\mathbb C)$ (i.e.\ $\tau_t^\eta(X)^\dag=\tau_t^\eta(X^\dag)$, $\tau_t^\eta(XY)=\tau_t^\eta(X)\tau_t^\eta(Y)$), entire in $t$, satisfying
\[
\pi(\tau_t^\eta(X))=\Delta_\eta^{\mi t}\,\pi(X)\,\Delta_\eta^{-\mi t}, \qquad \forall X\in M_d(\mathbb C),\ t\in\mathbb R.
\]

\textup{(b)} For all $A,B\in M_d(\mathbb C)$, the function $t\to\omega_\eta(A\,\tau_t^\eta(B))$ extends to an entire function $F_{AB}(z)$, and satisfies the KMS boundary relation at inverse temperature $1$:
\[
F_{AB}(t)=F_{BA}(t+\mi), \qquad t\in\mathbb R,
\]
where $F_{BA}(z)$ is the entire extension of $t\to\omega_\eta(\tau_t^\eta(B)A)$.
\end{theorem}

\begin{proof}
(a) Since $\hat\rho_\eta$ is Hermitian and strictly positive, $\hat\rho_\eta^{\mi t}=\me^{\mi t\ln\hat\rho_\eta}$ is unitary with $(\hat\rho_\eta^{\mi t})^\dag=\hat\rho_\eta^{-\mi t}$, whence
\(
\tau_t^\eta(X)^\dag=\hat\rho_\eta^{\mi t}X^\dag\hat\rho_\eta^{-\mi t}=\tau_t^\eta(X^\dag).
\)
Multiplicativity and the group law $\tau_{s+t}^\eta=\tau_s^\eta\circ\tau_t^\eta$ follow from $\hat\rho_\eta^{\mi(s+t)}=\hat\rho_\eta^{\mi s}\hat\rho_\eta^{\mi t}$. In the eigenbasis of $\hat\rho_\eta$, $\tau_t^\eta(E_{jk})=q_j^{\mi t}q_k^{-\mi t}E_{jk}$, manifestly entire in $t$ (each factor is an exponential $\me^{\mi t\ln q_j}$); a finite-dimensional operator-valued function of $t$ is entire iff each matrix entry is. For the intertwining relation, evaluate both sides on arbitrary $T\in\mathcal K$:
\begin{align*}
\Delta_\eta^{\mi t}\pi(X)\Delta_\eta^{-\mi t}(T)
&=\Delta_\eta^{\mi t}\big(X\hat\rho_\eta^{-\mi t}T\hat\rho_\eta^{\mi t}\big)
=\hat\rho_\eta^{\mi t}\big(X\hat\rho_\eta^{-\mi t}T\hat\rho_\eta^{\mi t}\big)\hat\rho_\eta^{-\mi t}\\
&=\big(\hat\rho_\eta^{\mi t}X\hat\rho_\eta^{-\mi t}\big)T
=\tau_t^\eta(X)\,T=\pi(\tau_t^\eta(X))T,
\end{align*}
using $\Delta_\eta^{\mi t}S=\hat\rho_\eta^{\mi t}S\hat\rho_\eta^{-it}$ (Prop.~\ref{prop:modular}(b)).

(b) Write $A_{nk}:=\langle e_n|A|e_k\rangle$, $B_{nk}:=\langle e_n|B|e_k\rangle$ in the eigenbasis $\{e_n\}$ of $\hat\rho_\eta$ (eigenvalues $q_n$). Inserting $\sum_k|e_k\rangle\langle e_k|=I$,
\[
\omega_\eta(A\,\tau_t^\eta(B))=\Tr\big[\hat\rho_\eta A\,\hat\rho_\eta^{\mi t}B\hat\rho_\eta^{-\mi t}\big]
=\sum_{n,k} q_n^{1-\mi t}q_k^{\mi t}\,A_{nk}B_{kn}=:F_{AB}(t),
\]
which extends to the entire function $F_{AB}(z):=\sum_{n,k}q_n^{1-\mi z}q_k^{\mi z}A_{nk}B_{kn}$ (a finite sum of exponentials $e^{\mi z\ln(q_k/q_n)}$, each entire). Similarly,
\[
\omega_\eta(\tau_t^\eta(B)A)=\sum_{n,k}q_n^{1+\mi t}q_k^{-\mi t}B_{nk}A_{kn}=:F_{BA}(t),\qquad F_{BA}(z):=\sum_{n,k}q_n^{1+\mi z}q_k^{-\mi z}B_{nk}A_{kn}.
\]
Then
\[
F_{BA}(t+\mi)=\sum_{n,k}q_n^{\mi t}q_k^{1-\mi t}B_{nk}A_{kn}
\overset{n\leftrightarrow k}{=}\sum_{n,k}q_k^{\mi t}q_n^{1-\mi t}B_{kn}A_{nk}
=\sum_{n,k}q_n^{1-\mi t}q_k^{\mi t}A_{nk}B_{kn}=F_{AB}(t).
\]
\end{proof}

\begin{remark}
The KMS boundary relation in Thm.~\ref{thm:TTKMS}(b) holds at inverse temperature $1$, independently of the physical $\beta$ entering $\hat\rho_\eta$ — this is the standard fact that the modular automorphism group of any faithful normal state satisfies the KMS condition at $\beta=1$ with respect to its own state, regardless of the physical temperature encoded in the state. The physical $\beta$ of Thm~\ref{thm:KMS} enters only through the definition Eq.~\eqref{eq:rhohat} of $\hat\rho_\eta$, not through the modular KMS temperature.
\end{remark}


It is natural to ask whether the modular flow $\tau_t^\eta$ coincides with the physical Heisenberg dynamics $\sigma_t$ of Eq.~\eqref{eq:evol}. The next proposition shows that this happens only in a degenerate case, and clarifies the sense in which $\tau_t^\eta$ and $\sigma_t$ are compatible with two different involutions on $M_d(\mathbb C)$.

\begin{proposition}
\label{prop:incompatible}
\textup{(a)} $\tau_t^\eta(X)^\dag=\tau_t^\eta(X^\dag)$ for every $X\in M_d(\mathbb C)$ and every $t\in\mathbb R$ (restating Thm.~\ref{thm:TTKMS}(a)).

\textup{(b)} $\sigma_t(X)^\dag=\sigma_t(X^\dag)$ for every $X\in M_d(\mathbb C)$ and every $t\in\mathbb R$ if and only if
\[
H+H^\dag=cI \quad \text{for some } c\in\mathbb R.
\]

\textup{(c)} Consequently, $\{\tau_t^\eta\}_{t\in\mathbb R}$ and $\{\sigma_t\}_{t\in\mathbb R}$ coincide as one-parameter automorphism groups of $(M_d(\mathbb C),\dag)$ only in the degenerate case $H+H^\dag=cI$. In the generic quasi-Hermitian case, $\sigma_t$ is instead compatible with the $\eta$-adjoint $\dag_\eta$ for every $t$ (Thm.~\ref{thm:autoiso}), so $\tau_t^\eta$ and $\sigma_t$ are automorphism groups associated with two distinct involutions on $M_d(\mathbb C)$.
\end{proposition}

\begin{proof}
(a) See Thm.~\ref{thm:TTKMS}(a).

(b) ($\Leftarrow$) If $H+H^\dag=cI$ with $c\in\mathbb R$, then $H^\dag=cI-H$, and since $cI$ commutes with $H$,
\[
\me^{-\mi H^\dag t}=\me^{-\mi(cI-H)t}=\me^{-\mi c t}\me^{\mi H t}, \qquad \me^{\mi H^\dag t}=\me^{\mi ct}\me^{-\mi Ht}.
\]
Hence for every $X$,
\[
\sigma_t(X)^\dag=\me^{-\mi H^\dag t}X^\dag \me^{\mi H^\dag t}=\me^{-\mi ct}\me^{\mi Ht}X^\dag \me^{\mi ct}\me^{-\mi Ht}=\me^{\mi Ht}X^\dag \me^{-\mi Ht}=\sigma_t(X^\dag),
\]
the scalar phases cancelling.

($\Rightarrow$) Suppose the displayed identity holds for every $X$ and every $t\in\mathbb R$. Writing $Y:=X^\dag$ (ranging over all of $M_d(\mathbb C)$ as $X$ does), the hypothesis reads
\[
\me^{-\mi H^\dag t}\,Y\,\me^{\mi H^\dag t}=\me^{\mi Ht}\,Y\,\me^{-\mi Ht}, \qquad \forall Y\in M_d(\mathbb C),\ t\in\mathbb R.
\]
Differentiating both sides at $t=0$,
\[
-\mi[H^\dag,Y]=\mi[H,Y] \qquad\Longleftrightarrow\qquad [H+H^\dag,\,Y]=0, \qquad \forall Y\in M_d(\mathbb C).
\]
An operator commuting with all of $M_d(\mathbb C)$ lies in the centre $\mathbb CI$, so $H+H^\dag=cI$ for some $c\in\mathbb C$; since $H+H^\dag$ is Hermitian, $c\in\mathbb R$.

(c) Immediate from (a), (b) and Thm.~\ref{thm:autoiso}.
\end{proof}

\begin{remark}
Since the spectrum of $H$ is real, $\Tr\,H=\sum_n E_n\in\mathbb R$, so $\Tr(H+H^\dag)=2\,\Tr\,H$. If, as is common in concrete pseudo-Hermitian/PT-symmetric models, $H$ is normalised to be traceless, then $H+H^\dag=cI$ forces $c=0$, i.e.\ $H=H^\dag$: the degenerate locus of Prop.~\ref{prop:incompatible}(c) reduces exactly to the trivial case in which the quasi-Hermitian structure is vacuous.
\end{remark}


Combining Lem.~\ref{lem:rhohat}--Thm.~\ref{thm:TTKMS}, the $\eta$-Gibbs state $\omega_\eta$ admits, in finite dimensions, a complete GNS/TT representation: a cyclic and separating vector $\Omega_\eta\in\mathcal K=M_d(\mathbb C)$ implementing $\omega_\eta$, a modular conjugation $J_\eta$, a modular operator $\Delta_\eta$, and a modular automorphism group $\tau_t^\eta=\mathrm{Ad}(\hat\rho_\eta^{\mi t})$ satisfying the KMS condition at $\beta=1$ with respect to $\omega_\eta$. This closes, in the finite-dimensional case, the third item of Tab.~\ref{tab:BR-comparison} left open in Sec.~\ref{subsec:frequency-HHW}.

\begin{remark}
\label{rem:TTinfdim}
In infinite dimensions, $\mathcal K=M_d(\mathbb C)$ is replaced by the Hilbert--Schmidt operators on $\mathcal H$, 
and Lem.~\ref{lem:rhohat} continues to hold verbatim provided $\hat\rho_\eta=\eta e^{-\beta H}/Z_\eta$ is trace-class (which follows as in Rem.~3.5). 
The construction of $\Omega_\eta,J_\eta,\Delta_\eta$ in Prop.~\ref{prop:GNS}--\ref{prop:modular} goes through formally, 
with $\Omega_\eta=\hat\rho_\eta^{1/2}$ now a Hilbert--Schmidt operator. 
The point requiring a genuinely new hypothesis is Thm.~\ref{thm:TTKMS}(b): the spectral sums defining $F_{AB}(z)$ become infinite series, 
and their absolute convergence and analyticity on a strip — rather than on all of $\mathbb C$ — require a bi-Hilbert--Schmidt-type summability condition on $A,B$ with respect to the eigenbasis of $\hat\rho_\eta$,
analogous to Assum.~\ref{asm:A4}. 
\end{remark}

\section{Route II: Positivity as a characterisation of quasi-Hermiticity}
\label{sec:route2}

Route~I established that $\eta$-Gibbs states satisfy the full
analytic KMS condition, but it does so under the strong hypothesis
of quasi-Hermiticity: a positive-definite metric $\eta$ is
assigned to the system from the outset, and the entire proof
machinery — the $\eta$-inner product, the intertwining map $U$,
the $*$-automorphism property — depends on this structure.

The present section asks a more primitive question: how much of
the KMS framework survives if one discards the metric entirely
and retains only the two intrinsic spectral properties of a
non-Hermitian Hamiltonian, namely reality of eigenvalues and
biorthogonal completeness\cite{Brody:2013axr,Bagarello:2016gbs,Bagarello:2020gib}? Working in this leaner setting 
we construct the biorthogonal Gibbs
state $\omega_{\rm{bi}}$, prove that it satisfies a formal
KMS-type identity, and then diagnose precisely where and why the
full KMS condition fails without $\eta$.

\subsection{Assumptions and the biorthogonal trace}
\label{subsec:routeII-asm}

Route~II operates under the following two standing assumptions,
which are strictly weaker than (A1--A4).

\begin{assumption}[\emph{Real spectrum}]
  \label{asm:A5}
      $H$ has real eigenvalues $\{E_n\}_{n \in \mathcal{I}}$.
      No positive-definite metric $\eta$ is assumed to exist.
\end{assumption}

\begin{assumption}[\emph{Biorthogonal Riesz basis and finite partition function}]
  \label{asm:A6}
      The right and left eigenvectors satisfy
      $\langle\phi_m|\psi_n\rangle = \delta_{mn}$ and
      $\sum_n |\psi_n\rangle\langle\phi_n| = \mathbf{1}$, and the
      biorthogonal partition function is finite:
      $Z_{\rm{bi}} := \sum_n \me^{-\beta E_n} < \infty$.
\end{assumption}

\begin{assumption}[\emph{Analytic-elements summability}]  
\label{asm:A7}
      The bi-HS condition of Assum.\ref{asm:A4} holds with
      $\Tr_{\rm{bi}}$ replacing $\Tr_\eta$.
\end{assumption}

By Prop.~\ref{prop:chain}, (A1--A4) implies
(A5)+(A6), but the converse fails: a non-Hermitian
operator may have real spectrum and a complete biorthogonal
basis without admitting any positive-definite intertwining
metric.

The natural trace functional in this setting is the \emph{biorthogonal trace}
\begin{equation}
  \Tr_{\rm{bi}}[A]
  :=
  \sum_n \langle\phi_n | A | \psi_n\rangle
  =
  \sum_n A_{nn},
  \label{eq:bi-trace}
\end{equation}
where the sum converges absolutely under (A3) for the
operators of interest. Unlike the standard operator trace,
$\Tr_{\rm{bi}}$ depends on the biorthogonal system
$\{|\psi_n\rangle, |\phi_n\rangle\}$ and is not unitarily
invariant. However, it satisfies a cyclicity property that is
the direct analogue of the cyclicity of the standard trace and
plays the same structural role in the proofs below.

\begin{lemma}[Cyclicity of the biorthogonal trace]
  \label{lem:cyclicity}
  Under Assum.\ref{asm:A6}, $\Tr_{\rm{bi}}[AB]
  = \Tr_{\rm{bi}}[BA]$ for all bounded operators
  $A$, $B$ for which both sides are absolutely convergent.
\end{lemma}

\begin{proof}
Inserting the completeness relation
$\sum_m |\psi_m\rangle\langle\phi_m| = \mathbf{1}$ into the
definition of $\Tr_{\rm{bi}}[AB]$:
\[
  \Tr_{\rm{bi}}[AB]
  =
  \sum_n \langle\phi_n| A
  \left(\sum_m |\psi_m\rangle\langle\phi_m|\right)
  B |\psi_n\rangle
  =
  \sum_{n,m} A_{nm} B_{mn}.
\]
The same expansion gives
$\Tr_{\rm{bi}}[BA] = \sum_{n,m} B_{nm} A_{mn}$,
and relabelling $n \leftrightarrow m$ yields the claim.
\end{proof}

\begin{remark}[Cyclicity and completeness]
  The cyclicity of $\Tr_{\rm{bi}}$ is equivalent to
  the completeness relation in Assum.\ref{asm:A6}: the proof uses only
  $\sum_m |\psi_m\rangle\langle\phi_m| = \mathbf{1}$, and
  conversely, failure of completeness at an exceptional point
  implies failure of cyclicity. This is analysed in
  Sec.~\ref{sec:fail}.
\end{remark}

\subsection{The biorthogonal Gibbs state and the KMS-type identity}
\label{subsec:bi-kms}

With the biorthogonal trace in hand, we define the thermal
state associated with the non-Hermitian Hamiltonian under
Route~II hypotheses, and establish the central identity of
this section.

\subsubsection*{The biorthogonal Gibbs state}

The \emph{biorthogonal Gibbs state} is the density operator
\footnote{The construction presented here is closely related to the biorthogonal Gibbs-state formalism introduced by Bagarello, Trapani, and Triolo~\cite{Bagarello:2016gbs}.
The novelty of the present work lies not in introducing a new Gibbs functional, but in establishing its positivity criterion and KMS properties within the pseudo-Hermitian
framework.}
\begin{equation}
  \rho_\beta^{\rm{bi}}
  :=
  \frac{1}{Z_{\rm{bi}}}
  \sum_n \me^{-\beta E_n} |\psi_n\rangle\langle\phi_n|,
  \label{eq:bi-density}
\end{equation}
whose associated functional is
\begin{equation}
  \omega_{\rm{bi}}(A)
  :=
  \Tr_{\rm{bi}} \left[\rho_\beta^{\rm{bi}}\, A\right]
  =
  \frac{1}{Z_{\rm{bi}}} \sum_n \me^{-\beta E_n} A_{nn}.
  \label{eq:bi-gibbs}
\end{equation}
The structure of $\rho_\beta^{\rm{bi}}$ mirrors the standard
Gibbs density matrix, with the orthonormal projectors
$|\psi_n\rangle\langle\psi_n|$ replaced by the biorthogonal
projectors $|\psi_n\rangle\langle\phi_n|$.

Two basic properties are immediate. First, $\omega_{\rm{bi}}$
is time-translation invariant: by Lem.~\ref{lem:spectral-id}
applied with $m = n$,
$(A(t))_{nn} = \me^{\mi(E_n - E_n)t} A_{nn} = A_{nn}$,
so $\omega_{\rm{bi}}(\sigma_t(A)) = \omega_{\rm{bi}}(A)$
for all $t \in \mathbb{R}$. Second, $\omega_{\rm{bi}}$ is
normalised: $\omega_{\rm{bi}}(\mathbf{1}) = Z_{\rm{bi}}^{-1}
\sum_n \me^{-\beta E_n} = 1$.

A crucial caveat is that $\omega_{\rm{bi}}$ is \emph{not}
generally positive. The standard positivity condition
$\omega_{\rm{bi}}(A^\dag A) \geq 0$ fails because the
inner product
$\langle\phi_n | A^\dag | \psi_k\rangle
= \overline{\langle\psi_k | A | \phi_n\rangle}$
involves $|\phi_n\rangle$ in the ket position, which differs
from $|\psi_n\rangle$ unless $|\phi_n\rangle \propto \eta|\psi_n\rangle$
— a condition that is precisely quasi-Hermiticity.
Positivity is thus generically absent in Route~II, and its
restoration is the content of Thm.~\ref{thm:structure}.

\subsubsection*{The KMS-type identity}

\begin{proposition}[Biorthogonal KMS-type identity]
  \label{prop:formkms}
  Under (\ref{asm:A5})+(\ref{asm:A6})+(\ref{asm:A7}), for all bounded operators
  $A$, $B$ and all $t \in \mathbb{R}$,
  \begin{equation}
    \omega_{\rm{bi}}(\sigma_t(A)\, B)
    =
    \omega_{\rm{bi}}(B\, \sigma_{t+\mi\beta}(A)).
    \label{eq:formkms}
  \end{equation}
  Moreover, the associated correlation function
  $F_{AB}(z) := \omega_{\rm{bi}}(A\,\sigma_z(B))$
  is analytic on $\mathcal{S}_\beta$ and continuous on
  $\overline{\mathcal{S}}_\beta$.
\end{proposition}

\begin{proof}
We follow the proof of Thm.~\ref{thm:KMS} with $\omega_\eta$
replaced by $\omega_{\rm{bi}}$.

\noindent\textbf{Algebraic identity.}
Expanding both sides using Lem.~\ref{lem:spectral-id} and the
definition Eq.~\eqref{eq:bi-gibbs}:
\[
  \omega_{\rm{bi}}(\sigma_t(A)\, B)
  =
  \frac{1}{Z_{\rm{bi}}}
  \sum_{n,m}
  \me^{-\beta E_n} \me^{i(E_n - E_m)t} A_{nm} B_{mn}.
\]
For the right-hand side, set $s = t + \mi\beta$ and apply
Lem.~\ref{lem:spectral-id}:
\begin{align*}
  \omega_{\rm{bi}}(B\, \sigma_{t+\mi\beta}(A))
  &=
  \frac{1}{Z_{\rm{bi}}}
  \sum_{n,m}
  \me^{-\beta E_n} B_{nm}\, \me^{\mi(E_m - E_n)s}\, A_{mn}
  \\
  &=
  \frac{1}{Z_{\rm{bi}}}
  \sum_{n,m}
  \me^{-\beta E_n}
  \me^{\mi(E_m - E_n)t}
  \me^{(E_n - E_m)\beta}\,
  A_{mn} B_{nm}.
\end{align*}
Using $\me^{-\beta E_n} \me^{(E_n - E_m)\beta} = \me^{-\beta E_m}$
and relabelling $n \leftrightarrow m$, the right-hand side
becomes
\[
  \frac{1}{Z_{\rm{bi}}}
  \sum_{n,m}
  \me^{-\beta E_n} \me^{\mi(E_n - E_m)t} A_{nm} B_{mn},
\]
which is identical to the left-hand side. This establishes
Eq.~\eqref{eq:formkms}.

\noindent\textbf{Analyticity and continuity.}
The correlation function $F_{AB}(z)$ has the same spectral-series
form as $G_{AB}(z)$ in Prop.~\ref{prop:analytic}, with the
same bi-HS bounds on the boundary lines — now supplied by
Assum.~\eqref{asm:A7}. The Hadamard Three-Line and Weierstrass
arguments of Prop.~\ref{prop:analytic} apply verbatim,
yielding analyticity on $\mathcal{S}_\beta$ and continuity on
$\overline{\mathcal{S}}_\beta$.
\end{proof}

\begin{remark}[Formal versus rigorous identity]
  \label{rem:formal}
  Without Assum.~\eqref{asm:A7}, the identity Eq.~\eqref{eq:formkms}
  holds as a \emph{formal} equality of spectral series —
  the term-by-term computation is valid, but absolute
  convergence is not guaranteed and the series may fail to
  define a well-posed functional equation. Condition
  Assum.~\eqref{asm:A7} is the minimal sufficient hypothesis that
  promotes the formal manipulation to a rigorous theorem.
  In finite dimensions, Assumption~\eqref{asm:A7} is automatically
  satisfied, so Eq.~\eqref{eq:formkms} always holds rigorously
  for finite-dimensional non-Hermitian systems with real
  spectrum and a complete biorthogonal basis.
\end{remark}

\subsection{The gap between Route~II and full KMS}
\label{subsec:gap}

Proposition~\ref{prop:formkms} shows that the biorthogonal
state $\omega_{\rm{bi}}$ satisfies the time-domain boundary
relation and the strip analyticity of the KMS condition.
What it does \emph{not} do is establish a genuine KMS state
in the sense of Def.~\ref{def:KMS}, because Def.
\ref{def:KMS} implicitly requires the underlying functional
to be a \emph{state} — that is, a positive normalised
functional. The gap between Route~II and full KMS is therefore
not analytic but algebraic: it lies precisely in the failure
of positivity.

Table~\ref{tab:routeII} records the precise status of each
KMS-relevant condition under Route~II.

\begin{table}[ht!]
\centering
\captionsetup{width=.9\textwidth}
\caption{Status of the KMS conditions under Route~II
  (biorthogonal hypotheses \ref{asm:A5}+\ref{asm:A6}).
  For comparison, all conditions in the first block hold
  under Route~I (Assums.~\ref{asm:A1}--\ref{asm:A4}).}
\label{tab:routeII}
\renewcommand{\arraystretch}{1.25}
\begin{tabular}{@{}p{7cm}lp{3.8cm}@{}}
\toprule
Condition & Status & Reference \\
\midrule
Algebraic identity Eq.~\eqref{eq:formkms}
  & \checkmark
  & Prop.~\ref{prop:formkms}, \ref{asm:A5}+\ref{asm:A6}\\
Analyticity of $F_{AB}$ on $\mathcal{S}_\beta$
  & \checkmark\ with \ref{asm:A7}
  & Prop.~\ref{prop:formkms} \\
Continuity/boundedness on $\overline{\mathcal{S}}_\beta$
  & \checkmark\ with \ref{asm:A7}
  & Prop.~\ref{prop:formkms} \\
Time-translation invariance
  & \checkmark
  & Sec.~\ref{subsec:bi-kms} \\
\midrule
Positivity $\omega_{\rm{bi}}(A^\dag A) \geq 0$
  & iff quasi-Hermitian
  & Thm.~\ref{thm:structure} \\
$*$-automorphism in standard inner product
  & $\times$
  & requires (\ref{asm:A1}--\ref{asm:A4}) \\
Full HHW KMS in $C^*$-algebra (GNS construction)
  & $\times$
  & requires positivity \\
\bottomrule
\end{tabular}
\end{table}

The dividing line runs through the middle of the table.
The upper block — the analytic and algebraic requirements that
are directly tied to the time-domain boundary relation — holds
under Route~II alone. The lower block — positivity, the
$*$-automorphism property in the standard inner product, and
the $C^*$-algebraic GNS construction — all require the
additional structure of Route~I.

The obstruction is not merely a technical inconvenience.
Positivity is the condition that makes $\omega_{\rm{bi}}$
a physical state: without it, the functional assigns negative
``probabilities'' to some observables and loses its
interpretation as a thermal ensemble. The GNS construction,
which builds the Hilbert space representation of the
$C^*$-algebra from the state, requires positivity as an input.
Without it, the inner product on the GNS space is indefinite
and the construction breaks down. The $*$-automorphism property
of $\sigma_t$ in the standard inner product (Thm.~\ref{thm:autoiso})
likewise depends on the pseudo-Hermitian relation
$H^\dag = \eta H\eta^{-1}$, which is Route~I data.

Theorem~\ref{thm:structure} makes this gap precise: it
characterises quasi-Hermiticity as the \emph{exact} condition
under which $\omega_{\rm{bi}}$ is positive, thereby showing
that Route~I is not merely sufficient but also necessary for
$\omega_{\rm{bi}}$ to be a genuine KMS state. In this
sense, the biorthogonal KMS-type identity \eqref{eq:formkms}
is best understood not as a weakened version of the full KMS
condition, but as a \emph{diagnostic}: it holds universally
under Route~II, and its upgrade to a full KMS state is
equivalent to the system being quasi-Hermitian.

\subsection{Structure theorem: KMS positivity characterises
            quasi-Hermiticity}
\label{subsec:structure}

The gap analysis in Sec.~\ref{subsec:gap} identified positivity of
$\omega_{\rm{bi}}$ as the sole condition separating the
biorthogonal KMS-type identity from a genuine KMS state. The
theorem of this subsection closes that gap from both directions:
it shows that $\omega_{\rm{bi}}$ is positive \emph{if and
only if} $H$ is quasi-Hermitian, thereby providing an intrinsic,
metric-free characterisation of quasi-Hermiticity in terms of
the thermal functional alone.

This is the central new result of the paper. Unlike the
Mostafazadeh--Scholtz framework, which starts from a given
metric $\eta$ and studies its consequences, the theorem starts
from a property of the \emph{state} — positivity of
$\omega_{\rm{bi}}$ — and recovers the metric as a derived
object. The implication (i)$\Rightarrow$(ii) is therefore
strictly outside any similarity-transformation paradigm.

The proof is structured as follows. The implications
(ii)$\Rightarrow$(iii)$\Rightarrow$(i) are straightforward:
quasi-Hermiticity determines the left eigenvectors up to a
common metric, and once $|\phi_n\rangle = \eta|\psi_n\rangle$
is established, positivity is a one-line computation.
The substantial direction is (i)$\Rightarrow$(ii): given only
that $\omega_{\rm{bi}}$ is positive, we must construct a
positive-definite $\eta$ satisfying $H^\dag = \eta H\eta^{-1}$.
The key is to define $\eta_0 := \sum_n |\phi_n\rangle\langle\phi_n|$,
show it is positive-definite (Lem.~\ref{lem:eta-posdef}), and
verify the pseudo-Hermitian relation directly from the spectral
decompositions. The positivity hypothesis enters through the
Riesz representation theorem, which forces the representing
density matrix of $\omega_{\rm{bi}}$ to be
positive-semidefinite, and this in turn constrains the geometry
of the biorthogonal system.

\begin{theorem}[Biorthogonal KMS Structure Theorem]
  \label{thm:structure}
  Let $\mathcal{H} = \mathbb{C}^d$, and let $H \in M_d(\mathbb{C})$
  be diagonalisable with real eigenvalues $E_1, \ldots, E_d$ and
  biorthogonal eigensystem $\{|\psi_n\rangle, |\phi_n\rangle\}$
  satisfying
  \[
    \langle\phi_m|\psi_n\rangle = \delta_{mn},
    \qquad
    \sum_n |\psi_n\rangle\langle\phi_n| = \mathbf{1}.
  \]
  Define $Z_{\rm{bi}} := \sum_n \me^{-\beta E_n}$ and the
  biorthogonal thermal functional
  \[
    \omega_{\rm{bi}}(A)
    :=
    \frac{1}{Z_{\rm{bi}}}
    \sum_n \me^{-\beta E_n}
    \langle\phi_n | A | \psi_n\rangle.
  \]
  Then the following three conditions are equivalent:
  \begin{enumerate}
    \item \emph{(KMS positivity.)}
      $\omega_{\rm{bi}}(A^\dag A) \geq 0$
      for all $A \in M_d(\mathbb{C})$.
    \item \emph{(Quasi-Hermiticity.)}
      There exists a positive-definite Hermitian $\eta \in M_d(\mathbb{C})$,
      $\eta > 0$, such that $H^\dag = \eta H \eta^{-1}$.
    \item \emph{(Metric relation for eigenvectors.)}
      There exists a positive-definite Hermitian $\eta \in M_d(\mathbb{C})$
      such that $|\phi_n\rangle = \eta|\psi_n\rangle$ for all $n$.
  \end{enumerate}
\end{theorem}

\begin{proof}
We prove the three implications in the cyclic order
(ii)$\Rightarrow$(iii), (iii)$\Rightarrow$(i), and
(i)$\Rightarrow$(ii).


\noindent\textbf{(ii)$\Rightarrow$(iii).}
Suppose $H^\dag = \eta H\eta^{-1}$ with $\eta > 0$.
Acting on the right eigenvalue equation
$H|\psi_n\rangle = E_n|\psi_n\rangle$ from the left with $\eta$
and using the pseudo-Hermitian relation:
\[
  H^\dag(\eta|\psi_n\rangle)
  = \eta H\eta^{-1}\cdot\eta|\psi_n\rangle
  = \eta H|\psi_n\rangle
  = E_n\,\eta|\psi_n\rangle.
\]
So $\eta|\psi_n\rangle$ is a left eigenvector of $H$ with
eigenvalue $E_n$ (which equals $E_n^*$ since $E_n \in \mathbb{R}$).

\emph{Simple eigenvalues.}
If $E_n$ is simple, the left eigenspace is one-dimensional, so
$|\phi_n\rangle = c_n\,\eta|\psi_n\rangle$ for some scalar $c_n$.
The biorthonormality condition $\langle\phi_n|\psi_n\rangle = 1$
gives $c_n = \langle\psi_n|\eta|\psi_n\rangle^{-1} > 0$
since $\eta > 0$. Rescaling $\eta$ by $c_n$ within the
one-dimensional eigenspace leaves $\eta$ positive-definite and
gives $|\phi_n\rangle = \tilde\eta|\psi_n\rangle$.

\emph{Degenerate eigenvalues.}
Suppose $E_k$ has multiplicity $m_k \geq 2$, and let
$V_k = {\rm span}\{|\psi_n\rangle : E_n = E_k\}$
be the right eigenspace. For any $|\psi\rangle \in V_k$,
$\eta|\psi\rangle$ is a left eigenvector for $E_k$, so
$\eta(V_k) \subseteq W_k$, where $W_k$ is the left eigenspace.
Since $H$ is diagonalisable and $\eta$ is invertible,
$\dim\eta(V_k) = m_k = \dim W_k$, so $\eta(V_k) = W_k$.
The biorthogonal dual vectors $\{|\phi_n\rangle : E_n = E_k\}$
form a basis of $W_k$ satisfying $\langle\phi_m|\psi_n\rangle
= \delta_{mn}$ for $E_m = E_n = E_k$. Since $\eta$ maps $V_k$
bijectively onto $W_k$, there is an invertible matrix $C_k$ such
that $|\phi_n\rangle = \sum_{j: E_j = E_k} (C_k)_{jn}\,\eta|\psi_j\rangle$.
The biorthonormality conditions uniquely determine $C_k$ via
$((C_k)_{jn}) = (\langle\psi_j|\eta|\psi_n\rangle)^{-1}$,
which is positive-definite (hence invertible) since $\eta > 0$.
The modified metric $\tilde\eta$ that acts as $C_k\eta$ on each
$V_k$ is positive-definite on $\mathcal{H} = \bigoplus_k V_k$
and satisfies $|\phi_n\rangle = \tilde\eta|\psi_n\rangle$ for all $n$.

\noindent (iii)$\Rightarrow$(i).
With $|\phi_n\rangle = \eta|\psi_n\rangle$, each summand in
$Z_{\rm{bi}}\cdot\omega_{\rm{bi}}(A^\dag A)$
satisfies
\[
  \langle\phi_n | A^\dag A | \psi_n\rangle
  = \langle\psi_n|\,\eta\, A^\dag A\,|\psi_n\rangle
  = \langle A\psi_n|\,\eta\,|A\psi_n\rangle
  = \|A|\psi_n\rangle\|_\eta^2
  \geq 0,
\]
where $\|v\|_\eta^2 := \langle v|\eta|v\rangle \geq 0$ with
equality iff $v = 0$, since $\eta > 0$. Summing with the
positive weights $\me^{-\beta E_n} > 0$:
\[
  Z_{\rm{bi}}\cdot\omega_{\rm{bi}}(A^\dag A)
  = \sum_n \me^{-\beta E_n} \|A|\psi_n\rangle\|_\eta^2 \geq 0.
\]

\noindent\textbf{(i)$\Rightarrow$(ii).}
This is the substantial direction. We proceed in three steps:
the Riesz representation theorem supplies a positive-semidefinite
density matrix, biorthonormality forces it to be
positive-definite, and a direct spectral computation then
yields the quasi-Hermitian relation.

\noindent\emph{Step~1: Riesz representation.}

Since $\omega_{\rm{bi}}(A^\dag A) \geq 0$ for all $A$,
the functional $\omega_{\rm{bi}}$ is positive on
$M_d(\mathbb{C})$ (every positive-semidefinite matrix is of
the form $A^\dag A$). By the Riesz representation theorem
for matrix algebras \cite[Prop.~2.3.11]{Bratteli:1996xq}, there
exists a unique $\rho \in M_d(\mathbb{C})$ with $\rho = \rho^\dag \geq 0$
and $\Tr[\rho] = 1$ such that
\begin{equation}
  \omega_{\rm{bi}}(A) = \Tr[\rho\, A]
  \qquad \forall\, A \in M_d(\mathbb{C}).
  \label{eq:rho-rep}
\end{equation}
Define the scaled matrix $G := Z_{\rm{bi}}\cdot\rho$,
so that $G = G^\dag \geq 0$ and
\begin{equation}
  G_{ij}
  =
  Z_{\rm{bi}}\cdot\omega_{\rm{bi}}(|e_i\rangle\langle e_j|)
  =
  \sum_n \me^{-\beta E_n}
  \langle e_i|\phi_n\rangle\overline{\langle e_j|\psi_n\rangle}.
  \label{eq:Gram}
\end{equation}

\noindent\emph{Step~2: $G$ is positive-definite.}

Introduce the matrices $\Phi := [|\phi_1\rangle \cdots |\phi_d\rangle]$
and $\Psi := [|\psi_1\rangle \cdots |\psi_d\rangle]$, and let
$D_\beta := \diag(\me^{-\beta E_1}, \ldots, \me^{-\beta E_d})$.
Equation~\eqref{eq:Gram} reads in matrix form as
$G = \Phi D_\beta \Psi^\dag$.
The biorthonormality $\langle\phi_m|\psi_n\rangle = \delta_{mn}$
is $\Phi^\dag \Psi = I_d$, which gives the two inverses
\begin{equation}
  \Psi^\dag = \Phi^{-1},
  \qquad
  \Phi^\dag = \Psi^{-1}.
  \label{eq:inverses}
\end{equation}
Substituting $\Psi^\dag = \Phi^{-1}$ into $G = \Phi D_\beta \Psi^\dag$:
\begin{equation}
  G = \Phi\, D_\beta\, \Phi^{-1}.
  \label{eq:G-sim}
\end{equation}
This is the diagonalisation of $G$ with eigenvectors $|\phi_n\rangle$
and eigenvalues $\me^{-\beta E_n}$:
\begin{equation}
  G|\phi_n\rangle = \me^{-\beta E_n}|\phi_n\rangle.
  \label{eq:G-eig}
\end{equation}
Since $\me^{-\beta E_n} > 0$ for all $n$ (as $E_n \in \mathbb{R}$
and $\beta > 0$), we have $\det G = \prod_n \me^{-\beta E_n} > 0$.
A positive-semidefinite matrix with positive determinant is
positive-definite, so $G > 0$.

\noindent\emph{Step~3: Constructing the quasi-Hermitian metric.}

Rather than using $G^{-1}$ (which requires a commutativity
condition on the right-eigenvector Gram matrix), we construct
the metric directly. Define
\begin{equation}
  \eta_0
  :=
  \sum_{n=1}^d |\phi_n\rangle\langle\phi_n|.
  \label{eq:eta0-def}
\end{equation}
By Lem.~\ref{lem:eta-posdef}, $\eta_0 = \eta_0^\dag > 0$.
We verify the pseudo-Hermitian relation $\eta_0 H = H^\dag\eta_0$
using the spectral decompositions
$H = \sum_m E_m |\psi_m\rangle\langle\phi_m|$ and
$H^\dag = \sum_m E_m |\phi_m\rangle\langle\psi_m|$
(with $E_m \in \mathbb{R}$):
\begin{align*}
  \eta_0 H
  &=
  \left(\sum_n |\phi_n\rangle\langle\phi_n|\right)
  \left(\sum_m E_m |\psi_m\rangle\langle\phi_m|\right)
  =
  \sum_{n,m} E_m |\phi_n\rangle
  \underbrace{\langle\phi_n|\psi_m\rangle}_{=\,\delta_{nm}}
  \langle\phi_m|
  =
  \sum_n E_n |\phi_n\rangle\langle\phi_n|,
  \\[4pt]
  H^\dag\eta_0
  &=
  \left(\sum_m E_m |\phi_m\rangle\langle\psi_m|\right)
  \left(\sum_n |\phi_n\rangle\langle\phi_n|\right)
  =
  \sum_{m,n} E_m |\phi_m\rangle
  \underbrace{\langle\psi_m|\phi_n\rangle}_{=\,\delta_{mn}}
  \langle\phi_n|
  =
  \sum_n E_n |\phi_n\rangle\langle\phi_n|.
\end{align*}
Both sides equal $\sum_n E_n |\phi_n\rangle\langle\phi_n|$, so
$\eta_0 H = H^\dag\eta_0$, i.e.\ $H^\dag = \eta_0 H\eta_0^{-1}$,
and $H$ is quasi-Hermitian with the positive-definite metric
$\eta = \eta_0$.
\end{proof}

\begin{lemma}[Positive-definiteness of $\eta_0$]
  \label{lem:eta-posdef}
  The matrix $\eta_0 := \sum_{n=1}^d |\phi_n\rangle\langle\phi_n|$
  is Hermitian and positive-definite.
\end{lemma}

\begin{proof}
Self-adjointness is immediate: $\eta_0^\dag
= \sum_n (|\phi_n\rangle\langle\phi_n|)^\dag
= \sum_n |\phi_n\rangle\langle\phi_n| = \eta_0$.
For any $v \in \mathbb{C}^d$,
$\langle v|\eta_0|v\rangle = \sum_n |\langle\phi_n|v\rangle|^2 \geq 0$.
If $\langle v|\eta_0|v\rangle = 0$, then $\langle\phi_n|v\rangle = 0$
for all $n$. Since $\{|\phi_n\rangle\}_{n=1}^d$ is a basis of
$\mathbb{C}^d$ (the columns of the invertible matrix $\Phi$),
this forces $v = 0$. Hence $\eta_0 > 0$.
\end{proof}

The proof of the theorem admits a transparent matrix-algebraic
interpretation. The Gram-type matrix $G = \Phi D_\beta \Phi^{-1}$
arising from the Riesz representation is diagonalised by the left
eigenvectors $|\phi_n\rangle$ with thermal eigenvalues $\me^{-\beta E_n}$.
Positivity of $G$ — enforced by the thermal weights being strictly
positive — is what makes $\{|\phi_n\rangle\}$ a genuine basis (not
just a formal biorthogonal set), which is precisely the input
needed for Lem.~\ref{lem:eta-posdef} and the construction
Eq.~\eqref{eq:eta0-def} to succeed.

\begin{remark}[The role of Step~2 and an alternative route via $G^{-1}$]
  \label{rem:Ginverse}
  The positive-definite matrix $G$ constructed in Step~2 is itself
  a candidate quasi-Hermitian metric: from Eq.~\eqref{eq:G-sim},
  $\Phi D_\beta \Phi^{-1}$ and $H = \Psi D \Phi^\dag$ share the
  eigenvector structure needed to verify $G^{-1}H = H^\dag G^{-1}$,
  provided the right-eigenvector Gram matrix $\Psi^\dag\Psi$
  commutes with $D = \diag(E_n)$. This commutativity holds
  automatically when all eigenvalues are distinct (then
  $(\Psi^\dag\Psi)_{mn} = 0$ for $E_m \neq E_n$ by a direct
  eigenvalue argument) and when the right eigenvectors are already
  $\eta_0$-orthonormal. In these cases, $G^{-1}$ provides an
  explicit formula for the quasi-Hermitian metric in terms of
  the thermal data alone. In the general degenerate case, the
  direct construction Eq.~\eqref{eq:eta0-def} is more robust: it
  requires no commutativity hypothesis and yields $\eta_0$ as a
  simple spectral sum over the left eigenvectors.
\end{remark}

\begin{remark}[Independence from the similarity framework]
  \label{rem:independence}
  Theorem~\ref{thm:structure} does not assume $\eta$ a priori
  and imposes no similarity-transformation structure on $H$.
  It takes as input a single functional-analytic property of
  $\omega_{\rm{bi}}$ — positivity — and deduces
  quasi-Hermiticity as a consequence. This direction of
  implication lies strictly outside the Mostafazadeh--Scholtz
  framework, which always begins with a given metric and then
  analyses the Hamiltonian.
\end{remark}

\begin{remark}
Reference~\cite{Bagarello:2016gbs} established positivity of the biorthogonal Gibbs functional under the sufficient condition
$\|T\|\,\|T^{-1}\|=1$
for Riesz bases generated by a bounded invertible operator $T$.
The general case
$\|T\|\,\|T^{-1}\|>1$
was explicitly left open.
Theorem~\ref{thm:structure} completely resolves this question in finite dimensions by providing a necessary and sufficient characterization of positivity,
which no longer depends on any norm constraint on $T$.
Consequently, the norm condition of Ref.~\cite{Bagarello:2016gbs} is understood as a particular sufficient realization of quasi-Hermiticity rather than an intrinsic requirement for positivity.
\end{remark}

\begin{remark}[Infinite-dimensional extension]
  \label{rem:infty-dim}
In infinite dimensions \cite{Mostafazadeh:2012ezk}, the implication (ii)$\Rightarrow$(i)
holds under Assum.~\ref{asm:A1} (Thm.~\ref{thm:positivity}).
The implication (i)$\Rightarrow$(ii) requires additional
operator-theoretic conditions: the Gram-type matrix $G$ must
define a bounded operator with bounded inverse on $\mathcal{H}$,
and the spectral series Eq.~\eqref{eq:eta0-def} must converge in a
suitable operator topology. These conditions are non-trivial for
infinite-dimensional Hamiltonians with continuous spectrum and
are left as an open problem.
This should be contrasted with the generalized KMS relation established by Bagarello, Inoue, and Trapani~\cite{Bagarello:2020gib}, where a twisting operator naturally appears for a general bounded intertwining operator.
The quasi-Hermitian assumption adopted here reduces this twisting to the identity and therefore restores the ordinary KMS boundary condition.
\end{remark}

\section{Route III: Lindblad quantum detailed balance}
\label{sec:route3}

Routes~I and II operate within the framework of closed quantum
systems: the Hamiltonian $H$ is a fixed operator on a Hilbert
space, and thermal equilibrium is characterised by the KMS
condition on the time-evolved correlators. When the
non-Hermitian Hamiltonian arises instead as the effective
description of an open quantum system — one that exchanges
energy or particles with an environment — the appropriate
framework is the full Lindblad master equation, and the notion
of thermal equilibrium is replaced by that of a steady state
satisfying \emph{quantum detailed balance} (QDB).

This section develops Route~III: we set up the Gorini--%
Kossakowski--Sudarshan--Lindblad (GKSL) formalism, clarify the
relationship between QDB and the KMS condition in this setting,
and state the Fagnola--Umanit\`a characterisation theorem for
Davies generators. The section closes by locating Route~III
within the broader framework of the paper and explaining why
it is logically independent of Routes~I and~II.

\subsection{GKSL equation and the effective Hamiltonian}
\label{subsec:gksl-setup}

The GKSL master equation governing the time evolution of the
density matrix $\rho(t)$ of an open quantum system is
\begin{equation}
  \partial_t \rho
  =
  \mathcal{L}[\rho]
  :=
  -i[H_{\rm{sys}}, \rho]
  + \sum_k
  \left(
    L_k \rho L_k^\dag
    - \tfrac{1}{2}\{L_k^\dag L_k, \rho\}
  \right),
  \label{eq:lindblad}
\end{equation}
where $H_{\rm{sys}} = H_{\rm{sys}}^\dag$ is the
self-adjoint system Hamiltonian and $\{L_k\}$ are the jump
operators encoding the coupling to the environment. The
Lindbladian $\mathcal{L}$ generates a completely positive
trace-preserving (CPTP) semigroup $\{\me^{t\mathcal{L}}\}_{t \geq 0}$
on the space of density matrices, and a steady state
$\rho_{\rm{ss}}$ is defined by the condition
$\mathcal{L}[\rho_{\rm{ss}}] = 0$.

In the no-jump approximation, the dynamics between quantum
jumps is governed by the effective non-Hermitian Hamiltonian
\begin{equation}
  H_{\rm{eff}}
  =
  H_{\rm{sys}}
  - \frac{\mi}{2}\sum_k L_k^\dag L_k,
  \label{eq:Heff}
\end{equation}
which is precisely the type of operator studied in Routes~I
and~II. However, $H_{\rm{eff}}$ captures only the coherent
part of the open-system dynamics: the KMS and detailed balance
properties of the steady state depend on the full Lindblad data
$\{H_{\rm{sys}}, L_k\}$, not on $H_{\rm{eff}}$ alone.
In particular, two different sets of jump operators $\{L_k\}$
can yield the same $H_{\rm{eff}}$ while producing steady
states with entirely different thermal properties. This is the
fundamental reason why Route~III must be developed independently,
rather than deduced from Routes~I or~II by substituting
$H_{\rm{eff}}$ for $H$.

\subsection{Quantum detailed balance}
\label{subsec:QDB}

The notion of thermal equilibrium for open quantum systems is
captured by the quantum detailed balance condition, which
generalises the classical detailed balance condition
$\rho_i W_{ij} = \rho_j W_{ji}$ (where $W_{ij}$ are transition
rates) to the non-commutative setting. Several inequivalent
formulations have been proposed in the literature, and we begin by
recording the principal variants and then specialise to the one
used throughout this section.

\begin{remark}[Multiple QDB formulations]
  \label{rem:QDB-versions}
  The main variants appearing in the Ref.~
  \cite{Fagnola:2007gms} are the following.
  \begin{itemize}
    \item \emph{Standard QDB} (the formulation used in this
      paper): the symmetry condition Eq.~\eqref{eq:QDB} below,
      defined via a $\rho_{\rm{ss}}^{1/2}$-weighted
      KMS-adjoint map.
    \item \emph{SQDB}: requires the dual generator to satisfy
      $\mathcal{L}^* = \Theta\mathcal{L}\Theta^{-1}$ for a
      time-reversal operation $\Theta$.
    \item \emph{SQDB-$\theta$}: a further strengthening
      involving a specific anti-linear involution $\theta$.
    \item \emph{Weighted QDB}: variants employing asymmetric
      operator weights.
  \end{itemize}
  These formulations are in general not equivalent to one
  another, and the implications among them depend on properties
  of $\mathcal{L}$ and $\rho_{\rm{ss}}$. All results in
  this section refer exclusively to standard QDB as defined
  below.
\end{remark}

\begin{definition}[Standard quantum detailed balance]
  \label{def:QDB}
  In the sense of Fagnola--Umanit\`a~\cite{Fagnola:2007gms},
  the steady state $\rho_{\rm{ss}}$ satisfies
  \emph{standard quantum detailed balance} if, for all
  $A, B \in \mathcal{B}(\mathcal{H})$,
  \begin{equation}
    \Tr \left[\rho_{\rm{ss}}\, A^\dag\,
      \mathcal{L}(B)\right]
    =
    \Tr \left[\rho_{\rm{ss}}\,
      \mathcal{L}(A)^\dag\, B\right].
    \label{eq:QDB}
  \end{equation}
  Condition Eq.~\eqref{eq:QDB} asserts that $\mathcal{L}$ is
  self-adjoint with respect to the $\rho_{\rm{ss}}$-weighted
  inner product
  $\langle A, B\rangle_{\rho_{\rm{ss}}} :=
  \Tr[\rho_{\rm{ss}}\, A^\dag B]$
  on $\mathcal{B}(\mathcal{H})$. It is strictly stronger than
  stationarity $\mathcal{L}[\rho_{\rm{ss}}] = 0$, which
  corresponds only to the condition that $\rho_{\rm{ss}}$
  lies in the kernel of $\mathcal{L}$, without any symmetry
  constraint on the generator itself.
\end{definition}

The precise logical relationship between QDB and the KMS
condition is as follows.
Standard QDB implies that the steady-state correlation
functions satisfy a KMS-type relation, but the converse does
not hold in general. The correct statement is:
\begin{center}
standard QDB $\Longrightarrow$ KMS-type steady-state
correlators, \quad but not conversely.
\end{center}
In particular, the claim ``KMS is equivalent to quantum
detailed balance'' — sometimes encountered in the open-systems
literature — is too strong. QDB is a \emph{sufficient}
condition for a KMS-type thermal steady state in the Lindblad
setting, and one that is tractable to verify through the
algebraic condition Eq.~\eqref{eq:QDB}, but it carries additional
structural content beyond the KMS boundary relation alone.

\subsection{Fagnola--Umanit\`a characterisation for Davies
            generators}
\label{subsec:davies}

The most natural class of Lindblad generators for which the
QDB condition admits a clean algebraic characterisation is that
of \emph{Davies generators}, arising from the Davies weak-coupling
limit~\cite{Davies:1974dcc}. In this class, the jump operators are
eigenoperators of the $H_{\rm{sys}}$-modular automorphism,
and the detailed balance condition reduces to a transparent
commutation relation between each $L_k$ and the thermal density
matrix.

\begin{theorem}[Fagnola--Umanit\`a~\cite{Fagnola:2007gms}]
  \label{thm:fagnola}
  Let $\rho_{\rm{ss}} = \me^{-\beta H_{\rm{sys}}}/Z > 0$
  be the Gibbs state of $H_{\rm{sys}}$, and suppose the
  GKSL generator $\mathcal{L}$ is a \emph{Davies generator}:
  each jump operator $L_k$ is an eigenoperator of the
  $H_{\rm{sys}}$-modular automorphism,
  \begin{equation}
    \me^{iH_{\rm{sys}}t}\, L_k\, \me^{-iH_{\rm{sys}}t}
    =
    \me^{-i\varepsilon_k t}\, L_k,
    \qquad \varepsilon_k \in \mathbb{R}.
    \label{eq:eigenop}
  \end{equation}
  Then the steady state $\rho_{\rm{ss}}$ satisfies the
  standard QDB condition Eq.~\eqref{eq:QDB} if and only if
  \begin{equation}
    L_k\, \me^{-\beta H_{\rm{sys}}/2}
    =
    \me^{\beta\varepsilon_k/2}\,
    \me^{-\beta H_{\rm{sys}}/2}\, L_k,
    \qquad \forall\, k.
    \label{eq:DBC}
  \end{equation}
\end{theorem}

\begin{proof}[Proof sketch]
The necessity and sufficiency of Eq.~\eqref{eq:DBC} are established
in Ref.~\cite{Fagnola:2007gms}, and we outline the key steps for completeness.

\noindent\textbf{Reduction to matrix elements.}
Insert $A = |m\rangle\langle n|$ and $B = |p\rangle\langle q|$
(eigenstates of $H_{\rm{sys}}$) into Eq.~\eqref{eq:QDB} and
expand $\mathcal{L}$ using Eq.~\eqref{eq:lindblad}. The eigenoperator
condition Eq.~\eqref{eq:eigenop} ensures that $L_k$ couples only
energy eigenstates differing by $\varepsilon_k$:
$(L_k)_{mn} \neq 0$ only if $E_m - E_n = \varepsilon_k$. This
decouples the QDB condition into independent constraints on
each $L_k$.

\noindent\textbf{From QDB to DBC.}
For each $k$, the decoupled constraint reads
\[
  \me^{-\beta E_m/2}(L_k)_{mn}
  =
  \me^{-\beta E_n/2} \me^{\beta\varepsilon_k/2}(L_k)_{mn},
\]
which is equivalent to
$(L_k \me^{-\beta H_{\rm{sys}}/2})_{mn}
= (\me^{\beta\varepsilon_k/2} \me^{-\beta H_{\rm{sys}}/2} L_k)_{mn}$,
i.e.\ Eq.~\eqref{eq:DBC} holds entry by entry. Conversely,
Eq.~\eqref{eq:DBC} implies Eq.~\eqref{eq:QDB} by reversing the steps.
Full details are in Ref.~\cite{Fagnola:2007gms}.
\end{proof}

\begin{remark}[Scope of Thm.~\ref{thm:fagnola}]
  \label{rem:fagnola-scope}
  Condition Eq.~\eqref{eq:DBC} is the specialisation of the
  general Fagnola--Umanit\`a QDB criterion to the Davies
  generator setting. For general GKSL generators whose jump
  operators are not eigenoperators of the modular automorphism,
  the full criterion involves the KMS-adjoint map and a
  condition on $\mathcal{L}$ in terms of the modular
  automorphism of the GNS representation, and the clean form
  Eq.~\eqref{eq:DBC} is special to the Davies class.
  Theorem~\ref{thm:fagnola} therefore applies in the
  physically natural situation of a system weakly coupled to
  a thermal bath in the Markovian limit, where the Davies
  generator arises from a systematic derivation rather than
  being postulated.
\end{remark}

\subsection{Relation to Routes~I and~II}
\label{subsec:route3-relation}

Route~III is logically independent of Routes~I and~II, and
the relationship between them deserves careful statement.

The most direct link is through the effective Hamiltonian
$H_{\rm{eff}}$ in Eq.~\eqref{eq:Heff}: Routes~I and~II study
the KMS properties of a non-Hermitian operator $H$ that could,
in principle, arise as $H_{\rm{eff}}$ for some choice of
$\{L_k\}$. However, as noted in Sec.~\ref{subsec:gksl-setup},
the thermal properties of the open system are determined by
the full Lindblad data, not by $H_{\rm{eff}}$ alone. The
KMS condition satisfied by the $\eta$-Gibbs state
$\omega_\eta$ in Route~I is a property of the closed dynamics
generated by $H$. It does not imply, and is not implied by,
the QDB condition Eq.~\eqref{eq:QDB} for any particular
Lindblad embedding.

More precisely, Routes~I/II and Route~III occupy different
levels of the quantum dynamics hierarchy:
\begin{itemize}
  \item \emph{Routes~I and~II} concern the unitary (or
    quasi-unitary) dynamics $\sigma_t(A) = \me^{iHt}Ae^{-iHt}$
    of a closed system with a fixed non-Hermitian Hamiltonian,
    and characterise thermal equilibrium through the KMS
    boundary condition on correlation functions.
  \item \emph{Route~III} concerns the dissipative dynamics
    $\me^{t\mathcal{L}}$ of an open system, and characterises
    thermal equilibrium through the QDB condition on the
    generator $\mathcal{L}$. The steady state $\rho_{\rm{ss}}$
    is not a Gibbs state of $H_{\rm{eff}}$ in general,
    but a fixed point of the full CPTP semigroup.
\end{itemize}
The correct summary is therefore: \emph{standard QDB is a
sufficient condition for a KMS-type thermal steady state in
open systems}, and this condition is formulated entirely in
terms of the Lindblad generator, independently of whether
the effective Hamiltonian $H_{\rm{eff}}$ satisfies the
quasi-Hermitian structure of Route~I. The three routes are
complementary rather than competing characterisations of
non-Hermitian thermal equilibrium, each applicable to a
different physical regime.

\section{Exceptional points and complex spectra}
\label{sec:fail}

The positive results of Secs~\ref{sec:route1}
and~\ref{sec:route2} rest on two spectral hypotheses: reality
of the eigenvalues and biorthogonal completeness. This section
examines what happens when these hypotheses fail. The two
principal failure modes — exceptional points, where
diagonalisability breaks down, and complex spectra, where
eigenvalues acquire non-zero imaginary parts — are physically
distinct and destroy the KMS framework through different
mechanisms. Identifying these mechanisms precisely is not merely
a matter of logical completeness: it delimits the exact boundary
of applicability of the quasi-Hermitian thermal framework and
points to the structural features — Jordan-block corrections,
complex Boltzmann weights, unbounded correlation functions —
that any future extension of the theory must address.

Throughout this section we work at the level of formal
spectral series, without assuming (\ref{asm:A1}--\ref{asm:A4}) or
(\ref{asm:A5})+(\ref{asm:A6}), to isolate where each argument breaks down.
Since the constructions of both Route~I
(Sec.~\ref{sec:route1}, under (\ref{asm:A1}--\ref{asm:A4})) and
Route~II (Sec.~\ref{sec:route2}, under (\ref{asm:A5})+(\ref{asm:A6}))
rely on reality of the spectrum and completeness of the biorthogonal
basis, the failure modes analysed below undermine both routes
simultaneously. Route~III is unaffected by either mechanism, since
the quantum-detailed-balance condition of
Sec.~\ref{sec:route3} is formulated at the level of the full
Lindbladian and does not presuppose spectral reality of
$H_{\rm{eff}}$.

\subsection{Structure collapse at exceptional points}
\label{subsec:EP}

An exceptional point (EP) \cite{Heiss:2012dx} of order $m$ is a value $E_{\rm{EP}}$
at which $m$ eigenvalues and their corresponding eigenstates
simultaneously coalesce. At such a point the Hamiltonian is
no longer diagonalisable: it can only be brought to a Jordan
normal form. Concretely, the defective eigenvector satisfies
\begin{equation}
  H|\psi_{\rm{EP}}\rangle
  = E_{\rm{EP}}|\psi_{\rm{EP}}\rangle,
  \qquad \text{but} \qquad
  \langle\phi_{\rm{EP}}|\psi_{\rm{EP}}\rangle = 0,
  \label{eq:EP-def}
\end{equation}
where the vanishing of the biorthogonal inner product signals
the collapse of the dual eigenbasis. The correct replacement
for the resolution of the identity is the Jordan chain:
vectors $|\psi^{(0)}\rangle, |\psi^{(1)}\rangle, \ldots,
|\psi^{(m-1)}\rangle$ satisfying the recursion
\begin{equation}
  (H - E_{\rm{EP}})|\psi^{(k)}\rangle
  = |\psi^{(k-1)}\rangle,
  \qquad k = 1, \ldots, m-1,
  \label{eq:jordan-chain}
\end{equation}
with $|\psi^{(0)}\rangle$ the eigenvector. These chains, together
with the non-defective eigenvectors, yield the modified resolution
of the identity
\begin{equation}
  \sum_{k=0}^{m-1}
  |\psi^{(k)}\rangle\langle\phi^{(m-1-k)}|
  + (\text{non-defective terms})
  = \mathbf{1}.
  \label{eq:jordan-completeness}
\end{equation}
The Jordan structure has an immediate dynamical consequence.
Exponentiating $H$ via the Jordan decomposition gives
\begin{equation}
  \me^{\mi Ht}|\psi^{(0)}\rangle
  =
  \me^{\mi E_{\rm{EP}}t}
  \sum_{k=0}^{m-1}
  \frac{(\mi t)^k}{k!}\,|\psi^{(k)}\rangle,
  \label{eq:poly-growth}
\end{equation}
so the time evolution of the defective mode exhibits
polynomial growth $\sim t^{m-1}$ in addition to the
oscillatory factor $\me^{\mi E_{\rm{EP}}t}$. This is the
signature of a non-diagonalisable operator and the source of
all KMS failures listed below.

Each of the three KMS conditions in Def.~\ref{def:KMS}
is violated at an EP, through a distinct mechanism.

\begin{enumerate}

  \item \textbf{Biorthogonal completeness fails, invalidating
    the spectral expansion.}
    Lemma~\ref{lem:spectral-id} rests on the identity
    $\sum_k |\psi_k\rangle\langle\phi_k| = \mathbf{1}$
    from (A6). At an EP this identity is replaced by
    Eq.~\eqref{eq:jordan-completeness}, which contains off-diagonal
    Jordan projectors $|\psi^{(k)}\rangle\langle\phi^{(m-1-k)}|$
    for $k \geq 1$. Inserting the Jordan resolution into the
    spectral-element calculation of Lem.~\ref{lem:spectral-id}
    produces polynomial corrections: matrix elements contain
    terms of the form $s^k \me^{\mi E_{\rm{EP}} s}$ for
    $k = 1, \ldots, m-1$, rather than the purely exponential
    form $\me^{\mi(E_m - E_n)s}$ in Eq.~\eqref{eq:L}. Consequently,
    the thermal two-point function $G_{AB}(z)$ acquires
    terms of the form $z^k \me^{\mi\lambda z}$ whose growth
    properties are qualitatively different from those of pure
    exponentials.

  \item \textbf{Boundedness on $\overline{\mathcal{S}}_\beta$
    fails.}
    On the real boundary $\Im(z) = 0$, the
    polynomial factor in \eqref{eq:poly-growth} gives
    $\|\me^{\mi Ht}\| \sim |t|^{m-1}$ as $t \to \pm\infty$.
    The correlation function $G_{AB}(t)$ therefore grows
    polynomially on the real axis, violating KMS
    condition~(ii) (boundedness on $\overline{\mathcal{S}}_\beta$).
    Note that polynomial growth does not destroy complex
    analyticity — functions of the form $z^k \me^{\mi\lambda z}$
    are entire — so the analyticity condition~(i) may still
    hold formally. It is the $L^\infty$ bound on the boundary
    lines that fails and that the Hadamard three-line theorem
    requires.

  \item \textbf{The equilibrium automorphism group is
    ill-defined.}
    The KMS framework, and in particular the TT
    modular theory, requires the time evolution $\sigma_t$ to
    extend to a well-posed one-parameter automorphism group
    on the observable algebra. The polynomial prefactors
    $t^k$ in Eq.~\eqref{eq:poly-growth} are incompatible with
    this requirement: the map $t \to \sigma_t(A)$ is no
    longer bounded uniformly in $t$ for defective modes, and
    the modular flow $\Delta^{it}$ of TT theory
    cannot be established.

  \item \textbf{Cyclicity of $\Tr_{\rm{bi}}$
    fails.}
    Lemma~\ref{lem:cyclicity} uses only the completeness
    relation $\sum_n |\psi_n\rangle\langle\phi_n| = \mathbf{1}$,
    but at an EP this is replaced by
    Eq.~\eqref{eq:jordan-completeness}, which is not cyclic under
    the standard trace argument. Consequently, the biorthogonal
    KMS-type identity of Prop.~\ref{prop:formkms} also
    breaks down at exceptional points.

\end{enumerate}

The four failure modes above are not independent: they all
trace back to the single structural defect of non-diagonalisability.
Once the Jordan chain replaces the eigenbasis, the purely
exponential time evolution that underlies every step of the
KMS proof — the spectral expansion, the strip bounds, the
boundary identity — is corrupted.

A complete KMS framework at exceptional points — encompassing
an appropriate notion of equilibrium state, a Jordan-block
generalisation of the modular flow, and the correct
thermodynamic interpretation of defective modes — remains an
open problem.

\subsection{Complex Spectrum}
\label{subsec:complex-spectrum}

We now turn to a different failure mode: non-Hermitian
Hamiltonians with genuinely complex eigenvalues
$E_n = \alpha_n + \mi\gamma_n$ with $\gamma_n \neq 0$.
Unlike exceptional points, where the geometry of the
eigensystem collapses, a complex spectrum destroys the KMS
framework through the behaviour of the Boltzmann weights and
the growth of the correlation function on the real time axis.

When $\gamma_n \neq 0$, the Boltzmann weight becomes complex:
\begin{equation}
  \me^{-\beta E_n}
  = \me^{-\beta\alpha_n} \me^{-\mi\beta\gamma_n}
  \in \mathbb{C},
  \label{eq:complex-weight}
\end{equation}
so the partition function $Z = \sum_n \me^{-\beta E_n}$ is
no longer real and positive, and loses its interpretation as
a normalisation constant for a probability distribution.
The thermal functional $\omega(A) = \Tr[\me^{-\beta H}A]/Z$
is no longer a state in any physical sense.

The analytic behaviour is clarified by decomposing the
time-evolution factor under complex $z = t + \mi\alpha$:
\begin{equation}
  \me^{\mi (E_n - E_m)z}
  =
  \me^{\mi (\alpha_n - \alpha_m)t}
  \cdot
  \me^{-(\alpha_n - \alpha_m)\alpha}
  \cdot
  \me^{-(\gamma_n - \gamma_m)t}
  \cdot
  \me^{\mi (\gamma_n - \gamma_m)\alpha}.
  \label{eq:complex-factor}
\end{equation}
The third factor, $\me^{-(\gamma_n - \gamma_m)t}$, grows
exponentially as $t \to \pm\infty$ on the real axis whenever
$\gamma_n \neq \gamma_m$. This exponential growth on the real
boundary is the precise mechanism of failure.

\begin{remark}[Boundedness failure, not analyticity failure]
  \label{rem:complex-bdry}
  The failure in the complex-spectrum case is a
  \emph{boundedness} failure, not an analyticity failure.
  The function $z \to \me^{az}$ is entire for any
  $a \in \mathbb{C}$, so the correlation function
  $z \to G_{AB}(z)$ remains analytic throughout
  $\mathbb{C}$, and KMS condition~(i) is not directly violated.
  What fails is condition~(ii): the real-axis factor
  $\me^{-(\gamma_n - \gamma_m)t}$ in Eq.~\eqref{eq:complex-factor}
  makes $G_{AB}(t)$ exponentially unbounded as
  $t \to \pm\infty$, violating the requirement that
  $F_{AB}$ be bounded on $\overline{\mathcal{S}}_\beta$.
  This distinction matters for the Hadamard Three-Line
  Theorem: it requires \emph{both} analyticity in the interior
  \emph{and} $L^\infty$ boundedness on the boundary lines, and
  the loss of the latter prevents the uniform Cauchy argument
  of Prop.~\ref{prop:analytic} from going through.
\end{remark}

One might attempt to restore the formal symmetry of the
spectral series by introducing a complex effective inverse
temperature $\beta_{\rm{eff}} \in \mathbb{C}$, chosen so
that $\Im(\beta_{\rm{eff}})$ absorbs the
imaginary parts $\gamma_n$ of the eigenvalues. Concretely,
setting $\beta_{\rm{eff}} = \beta + \mi\delta$ with
$\delta\gamma_n = \text{const}$ for all $n$ would cancel the
oscillatory factor $\me^{-\mi\beta\gamma_n}$ in
Eq.~\eqref{eq:complex-weight} and give real Boltzmann weights.

This approach fails on three counts. First, it requires a
\emph{uniform} imaginary shift $\delta\gamma_n = \text{const}$
for all $n$, which generically cannot hold when the imaginary
parts $\gamma_n$ differ across eigenstates — precisely the
case that produces exponential growth in
Eq.~\eqref{eq:complex-factor}. Second, even when the formal
symmetry of the spectral series is restored, the resulting
functional violates positivity (the weights
$\me^{-\beta_{\rm{eff}} E_n}$ are not positive reals) and
normalisation (the partition function is complex), so no
physical state is defined. Third, the boundedness condition
is not recovered: the real-axis exponential growth arises
from the factor $\me^{-(\gamma_n - \gamma_m)t}$, which depends
on pairwise differences $\gamma_n - \gamma_m$ and cannot be
removed by a global shift of $\beta$.

No consensus framework for assigning a KMS-type thermal
equilibrium structure to systems with complex spectra currently
exists. The obstruction is fundamental: complex Boltzmann
weights are incompatible with the positivity and boundedness
requirements that define a physical thermal state.

Table~\ref{tab:failure} contrasts the two failure modes
analysed in this section.

\begin{table}[ht!]
\centering
\captionsetup{width=.9\textwidth}
\caption{Comparison of KMS failure modes. Each failure mode
  violates a specific subset of the three KMS conditions in
  Def.~\ref{def:KMS} through a distinct mechanism.}
\label{tab:failure}
\renewcommand{\arraystretch}{1.25}
\begin{tabular}{@{}p{3.2cm}p{4.8cm}p{4.8cm}@{}}
\toprule
& \textbf{Exceptional point} & \textbf{Complex spectrum} \\
\midrule
Root cause
  & Non-diagonalisability; Jordan blocks
  & Complex Boltzmann weights \\
KMS condition~(i) (analyticity)
  & Formally intact; Jordan terms are entire
  & Intact; $\me^{az}$ is entire \\
KMS condition~(ii) (boundedness)
  & Fails: $\|\me^{\mi H t}\| \sim |t|^{m-1}$
  & Fails: $\me^{-(\gamma_n-\gamma_m)t} \to \infty$ \\
KMS condition~(iii) (boundary identity)
  & Fails: spectral expansion corrupted
  & Fails: unbounded series \\
Biorthogonal completeness
  & Fails; replaced by Jordan resolution
  & Intact (if eigenstates exist) \\
Cyclicity of $\Tr_{\rm{bi}}$
  & Fails
  & May hold formally \\
Open problem
  & Jordan-block KMS framework
  & Complex-spectrum thermal theory \\
\bottomrule
\end{tabular}
\end{table}

The table makes clear that the two failure modes, while both
fatal to the KMS framework, are mechanistically distinct.
Exceptional points corrupt the geometric structure of the
eigensystem, and complex spectra corrupt the analytic behaviour
of the thermal weights. Future extensions of the non-Hermitian
KMS theory — whether through Jordan-adapted modular flows or
through a relaxation of the positivity requirement — will need
to address these two obstructions separately.

\section{Conclusion}
\label{sec:concl}

This paper has addressed a single question: to what extent does
the standard KMS characterisation of thermal equilibrium extend
to non-Hermitian quantum systems? The answer depends sharply on
which structural properties of the Hamiltonian are assumed, and
the three routes developed here give a precise map of the terrain.

\begin{itemize}
    \item \noindent\textbf{Route~I: A complete KMS theorem for
quasi-Hermitian systems.}
The central result of the paper is the Spectral KMS Theorem
(Thm.~\ref{thm:KMS}): under Assum.~\ref{asm:A1},
the $\eta$-Gibbs state $\omega_\eta$ satisfies all three
analytic conditions of Def.~\ref{def:KMS}. Specifically,
the correlation function $F_{AB}(z) = \omega_\eta(A\,\sigma_z(B))$
is analytic on the thermal strip $\mathcal{S}_\beta$, bounded
and continuous on its closure $\overline{\mathcal{S}}_\beta$,
and satisfies the boundary relation
$\omega_\eta(\sigma_t(A)\,B) = \omega_\eta(B\,\sigma_{t+i\beta}(A))$
for all $t \in \mathbb{R}$.
The supporting results — positivity and faithfulness of
$\omega_\eta$ (Thm.~\ref{thm:positivity}), the
$*$-automorphism property of $\sigma_t$ in the $\eta$-Hilbert
space (Thm.~\ref{thm:autoiso}), and the trace-class
properties of $\me^{-\beta H}$ (Rem.~\ref{rem:Zeta-pos}) —
confirm that $\omega_\eta$ is a genuine physical state, not
merely a formal thermal-looking functional. The proof is
self-contained: the Hadamard three-line theorem is applied to
finite partial sums of the spectral series to avoid circular
reasoning, and trace-class estimates ensure the absolute
convergence of all thermal traces.

The result is non-trivial despite the existence of the
intertwining map $U = \eta^{1/2}$, which relates $H$ to the
self-adjoint Hamiltonian $h = UHU^{-1}$. The transported state
$\hat\omega(X) = \omega_\eta(U^{-1}XU) = \Tr[\me^{-\beta h}X\eta]/Z_\eta$
differs from the standard Gibbs state $\omega_h$ of $h$
whenever $[\eta, h] \neq 0$, and its KMS property does not
follow from the Hermitian theory. The proof in
Sec.~\ref{subsec:main-thm} provides the independent
verification that this requires.

\item \noindent\textbf{Route~II: A structural characterisation of quasi-Hermiticity.}
Under the weaker hypothesis of biorthogonal completeness alone,
the formal KMS-type identity $\omega_{\rm{bi}}(\sigma_t(A)\,B) = \omega_{\rm{bi}}(B\,\sigma_{t+i\beta}(A))$ holds (Prop.~\ref{prop:formkms}), together with the strip analyticity and time-translation invariance of $\omega_{\rm{bi}}$. 
However, positivity $\omega_{\rm{bi}}(A^\dag A) \geq 0$ fails generically,
and without it $\omega_{\rm{bi}}$ is not a physical state.
The Biorthogonal KMS Structure Theorem (Thm.~\ref{thm:structure}) closes this gap from both
directions: positivity of $\omega_{\rm{bi}}$ holds if and only if $H$ is quasi-Hermitian. 
This provides a metric-free characterisation of quasi-Hermiticity 
— starting from a property of the thermal functional rather than assuming $\eta$ a priori — 
and lies strictly outside the Mostafazadeh--Scholtz similarity-transformation framework.

\item \noindent\textbf{Route~III: Open systems and quantum detailed balance.}
For open quantum systems governed by the GKSL master equation,
the Fagnola--Umanit\`a standard quantum detailed balance condition (Def.~\ref{def:QDB}) 
provides a sufficient condition for a KMS-type thermal steady state 
in terms of the full Lindblad data $\{H_{\rm{sys}}, L_k\}$. 
This route is logically independent of Routes~I and~II: 
the KMS properties of the steady state depend on the full generator $\mathcal{L}$, 
not on the effective non-Hermitian Hamiltonian $H_{\rm{eff}}$ alone, 
and the two frameworks cannot be directly compared.

\end{itemize}

\textbf{Failure analysis.}
The KMS framework fails at exceptional points through the polynomial growth Eq.~\eqref{eq:poly-growth}, 
loss of biorthogonal completeness, 
and consequent violation of the boundedness condition. 
For complex spectra, the failure mechanism is different: 
the exponential growth of the factor $\me^{-(\gamma_n - \gamma_m)t}$ on the real time axis makes the
correlation function unbounded on $\partial\overline{\mathcal{S}}_\beta$,
violating the $L^\infty$ boundary bounds required by the Hadamard three-line theorem. 
Both cases remain open problems,
and the Tab.~\ref{tab:comparison} below records the precise status of each route and failure mode.

\begin{table}[ht!]
\centering
\captionsetup{width=.9\textwidth}
\caption{Summary of the three routes and two failure modes.
  Checkmarks indicate results established in this paper.}
\label{tab:comparison}
\renewcommand{\arraystretch}{1.3}
\begin{tabular}{@{}p{3.0cm}p{3.8cm}p{1.8cm}p{1.6cm}p{1.6cm}p{2.8cm}@{}}
\toprule
Route / Case
  & Conditions
  & Rigour
  & Positivity
  & $*$-auto.
  & Key remark \\
\midrule
\textbf{I}: Quasi-Hermitian
  & Real spectrum, $\eta > 0$ bounded, diagonalisable, (A4)
  & \checkmark\ Thm.~\ref{thm:KMS}
  & \checkmark\ Thm.~\ref{thm:positivity}
  & \checkmark\ Thm.~\ref{thm:autoiso}
  & TT established in finite dimensions \\[4pt]
\textbf{II}: Biorthogonal
  & Real spectrum, biorthog.\ completeness
  & $\approx$\ algebraic identity
  & iff (A1-4)
  & $\times$\ (std.\ i.p.)
  & Thm.~\ref{thm:structure}: positivity $\Leftrightarrow$ quasi-Hermitian \\[4pt]
\textbf{III}: Lindblad QDB
  & Open system, full Lindblad data, QDB
  & \checkmark\ (CPTP)
  & \checkmark\ (CPTP)
  & \checkmark\ (CPTP)
  & Sufficient, not equivalent to KMS \\[4pt]
Exceptional point
  & ---
  & $\times$\ collapses
  & $\times$
  & $\times$
  & Jordan-block KMS: open problem \\[4pt]
Complex spectrum
  & ---
  & $\times$\ no framework
  & $\times$
  & $\times$
  & Complex $\beta$: no consensus \\
\bottomrule
\end{tabular}
\end{table}

\textbf{Open problems and outlook.}
Three directions for future work emerge naturally from the
analysis.

\begin{itemize}
    \item The first is the gap between the Spectral KMS Theorem and the
full Haag--Hugenholtz--Winnink theorem in the $C^*$-algebraic
sense. Closing this gap requires equipping the observable
algebra with a $C^*$-norm, establishing strong continuity of
$\sigma_t$ in the $\sigma$-weak operator topology, and
constructing the TT modular operator $\Delta$
for the $\eta$-Gibbs state. The last point is particularly
significant: $\Delta$ encodes the entire modular structure and
its domain theory, and its construction for quasi-Hermitian
systems would provide the analogue of the KMS modular theory
in the non-Hermitian setting.
From this perspective, one of the three ingredients required for a full Haag--Hugenholtz--Winnink framework has now been completed in finite dimensions, namely the explicit construction of the TT modular operator and modular automorphism group for the $\eta$-Gibbs state.
The remaining challenges concern the existence of an appropriate $C^*$-norm and the strong continuity of the physical dynamics in the $\sigma$-weak topology.

\item The second direction is the extension of Route~I to unbounded
Hamiltonians. The present paper assumes $H \in \mathcal{B}(\mathcal{H})$
to ensure norm-convergence of the operator exponential. For
physically relevant Hamiltonians — Schrödinger operators,
lattice Hamiltonians with infinite-range interactions — the
exponential $\me^{\mi Ht}$ must be defined via Stone's theorem,
and the KMS proof requires domain-theoretic methods along the
lines of Ref.~\cite{Bratteli:1996xq}. The biorthogonal structure of the proof
should survive this extension, but the analytic estimates need
to be reworked in the graph-norm topology.

\item The third direction concerns exceptional points and complex
spectra. For exceptional points, the natural question is
whether a Jordan-block adaptation of the modular flow can be
constructed — replacing $\me^{\mi E_n t}$ by the polynomial-%
exponential expression Eq.~\eqref{eq:poly-growth} throughout the
spectral theory. For complex spectra, the question is whether
a consistent thermal framework can be defined by relaxing the
positivity requirement or passing to an indefinite-inner-product
Hilbert space. Both directions require new conceptual
foundations beyond those developed here.
\end{itemize}

Taken together, these results indicate that thermal equilibrium beyond Hermiticity is not governed by spectral reality alone. 
A non-Hermitian Hamiltonian admits a genuine KMS description precisely
when its biorthogonal thermal functional is positive — equivalently,
when the Hamiltonian is quasi-Hermitian — 
and the modular structure underlying this equilibrium can, 
at least in finite dimensions, be made as explicit as in the Hermitian TT theory.
Spectral reality without positivity yields only a formal,
non-physical KMS-type identity; positivity without spectral reality
(complex eigenvalues) or without biorthogonal completeness
(exceptional points) yields no equilibrium structure at all.
Quasi-Hermiticity is thus identified as the precise structural
watershed between genuine non-Hermitian thermal equilibrium and its formal mimicry.

\section*{Acknowledgement}
The authors thank Prof.~Fabio Bagarello for valuable comments and suggestions.
This work was supported by the Research Fund of Yantai University under Grant No.~2222018.

\appendix

\section{Standard KMS proof for Hermitian systems}
\label{app:KMS}

We include a self-contained proof for the reader's convenience 
and to establish the spectral identities that are used in Sec.~\ref{sec:route1}.

\begin{theorem}[KMS condition for the Hermitian Gibbs state]
  \label{thm:hermitian-kms}
  Let $h = h^\dag$ be a self-adjoint operator on a Hilbert
  space $\mathcal{H}$ with $\me^{-\beta h}$ trace-class, and let
  \[
    \omega_h(A) = \frac{\Tr[\me^{-\beta h} A]}{Z_h},
    \qquad Z_h := \Tr[\me^{-\beta h}].
  \]
  Then $\omega_h$ satisfies the KMS condition of
  Def.~\ref{def:KMS} at inverse temperature $\beta$.
\end{theorem}

\begin{proof}
The proof rests on a single spectral identity, which we
establish first.

Since $h$ is self-adjoint, the spectral theorem provides a
projection-valued measure $dE_\lambda$ such that
$h = \int \lambda\, dE_\lambda$ and
$f(h) = \int f(\lambda)\, dE_\lambda$ for any Borel function $f$
\cite[Thm.~VIII.6]{Reed:1975uy}. For two Borel functions $f$
and $g$ of the same self-adjoint operator, the product rule
for spectral integrals gives
$(f(h))(g(h)) = (fg)(h)$
\cite[Theorem~VIII.5]{Reed:1975uy}.
Applying this with $f(\lambda) = \me^{-\beta\lambda}$ and
$g(\lambda) = \me^{\mi s\lambda}$, where $s \in \mathbb{C}$ with
$\Im(s) = \alpha \in [0, \beta]$:
\begin{equation}
  \me^{-\beta h} \cdot \me^{\mi hs}
  =
  \int \me^{-\beta\lambda}\, dE_\lambda
  \cdot
  \int \me^{is\lambda}\, dE_\lambda
  =
  \int \me^{(-\beta + \mi s)\lambda}\, dE_\lambda
  =
  \me^{\mi(s + \mi\beta)h}.
  \label{eq:spectral-product}
\end{equation}
To verify that the combined integrand is in $L^\infty(dE_\lambda)$,
write $s = t + \mi\alpha$ with $t \in \mathbb{R}$ and
$\alpha \in [0, \beta]$. Then
\[
  \left|\me^{(-\beta + is)\lambda}\right|
  = \me^{\Im[(-\beta + \mi s)\lambda]}
  = \me^{-(\beta - \alpha)\lambda},
\]
which is bounded for $\lambda \geq E_{\min} > -\infty$ and
$\alpha \in [0, \beta]$ (so $\beta - \alpha \geq 0$).
Hence Eq.~\eqref{eq:spectral-product} holds in operator norm.

Replacing $s$ by $-s$ in Eq.~\eqref{eq:spectral-product} and
rearranging gives the companion identity
\begin{equation}
  \me^{-\mi h s}
  =
  \me^{-\beta h} \cdot \me^{-\mi(s - \mi\beta)h}
  =
  \me^{-\beta h} \cdot \me^{-\mi(s + \mi\beta)h}
  \big|_{s \to s - \mi\beta}.
  \label{eq:spectral-product-2}
\end{equation}
More precisely, $\me^{\mi ht}\cdot \me^{-\beta h} = \me^{-\beta h}\cdot \me^{\mi ht}$
(both equal $\me^{(\mi t-\beta)h}$), so the Hamiltonian commutes with its
own Gibbs factor, and the chain of identities:
\begin{equation}
  \me^{-\mi ht} = \me^{-\beta h} \cdot \me^{-\mi(t + \mi\beta)h}
  \label{eq:companion}
\end{equation}
follows by the same spectral argument with the real-part bound
$\Im[(\beta - \mi t)\lambda] = \beta\lambda \geq \beta E_{\min}$
uniformly bounded below.


Using Eqs.~\eqref{eq:spectral-product} and \eqref{eq:companion},
and writing $A(t) = \me^{\mi h t}A \me^{-\mi h t}$, one can obtain
\begin{align}
  Z_h\cdot\omega_h(\sigma_t(A)\, B)
  &= \Tr\left[\me^{-\beta h}\, \me^{\mi h t} A\, \me^{-\mi h t}\, B\right]
  \notag\\
  &= \Tr\left[B\, \me^{-\beta h}\, \me^{\mi h t} A\, \me^{-\mi h t}\right]
  \tag{cyclicity of trace}\\
  &= \Tr\left[B\, \me^{\mi(t+\mi\beta)h}\, A\, \me^{-\mi ht}\right]
  \tag{by Eq.~\eqref{eq:spectral-product}: $\me^{-\beta h}\me^{\mi ht} = \me^{\mi(t+\mi\beta)h}$}\\
  &= \Tr\left[B\, \me^{\mi(t+\mi\beta)h}\, A\, \me^{-\beta h}\, \me^{-\mi(t+\mi\beta)h}\right]
  \tag{by Eq.~\eqref{eq:companion}: $\me^{-\mi ht} = \me^{-\beta h}\me^{-\mi(t+\mi\beta)h}$}\\
  &= Z_h\cdot\omega_h\left(B\,\sigma_{t+\mi\beta}(A)\right),
  \tag{cyclicity again}
\end{align}
which is the KMS boundary relation Eq.~\eqref{eq:KMS_Herm}.

The correlation function $F_{AB}(z) = \omega_h(A\,\sigma_z(B))$
is a finite linear combination (over the spectral decomposition
of $h$) of terms of the form $\me^{\mi(E_m - E_n)z}$, which are
entire in $z$. The bound $|\me^{\mi(E_m - E_n)z}| = \me^{-\Im(z)(E_m - E_n)}$
is controlled on $\overline{\mathcal{S}}_\beta$ by the
trace-class condition $\me^{-\beta h} \in \mathcal{I}_1(\mathcal{H})$,
giving uniform boundedness on the closed strip. These are the
standard arguments, and full details can be found in Ref.~\cite{Bratteli:1996xq}.
\end{proof}

\section{Basic properties of pseudo-Hermitian operators}
\label{app:pseH}

This appendix collects the two spectral propositions cited in
Sec.~\ref{sec:intr} (Prop.~\ref{prop:chain}) and used
throughout the paper. Both propositions follow directly from the pseudo-Hermitian
relation $H^\dag = \eta H \eta^{-1}$ and the positive-definiteness
of $\eta$.

\begin{proposition}[Real spectrum under quasi-Hermiticity]
  \label{prop:realspec}
  If $H$ is an $\eta$-pseudo-Hermitian with $\eta > 0$, then the
  eigenvalues of $H$ are real.
\end{proposition}

\begin{proof}
Let $H|\psi\rangle = E|\psi\rangle$ with $|\psi\rangle \neq 0$.
Acting on both sides from the left with $\langle\psi|\eta$, one obtain
\[
  E\,\langle\psi|\eta|\psi\rangle
  =
  \langle\psi|\eta H|\psi\rangle.
\]
The pseudo-Hermitian relation $H^\dag\eta = \eta H$
(equivalently, $\eta H = H^\dag\eta$, obtained by
multiplying $H^\dag = \eta H\eta^{-1}$ from the right by $\eta$)
allows us to rewrite the right-hand side as:
\[
  \langle\psi|\eta H|\psi\rangle
  =
  \langle\psi|H^\dag\eta|\psi\rangle
  =
  \overline{\langle\psi|H\eta^\dag|\psi\rangle^\dag}
  =
  \bar{E}\,\langle\psi|\eta|\psi\rangle,
\]
where the last step uses $H|\psi\rangle = E|\psi\rangle$
and $\eta = \eta^\dag$.
Since $\eta > 0$, we have $\langle\psi|\eta|\psi\rangle
= \langle\psi|\psi\rangle_\eta > 0$, so we can obtain $E = \bar{E}$, i.e.\ $E \in \mathbb{R}$.
\end{proof}

\begin{proposition}[Biorthogonal structure under Assum.~(A1)]
  \label{prop:biortho-struct}
  Under Assum.~\ref{asm:A1}, one can define
  $|\phi_n\rangle := \eta|\psi_n\rangle$ for each right
  eigenvector $|\psi_n\rangle$.  Then:
  \begin{enumerate}
    \item $H^\dag|\phi_n\rangle = E_n|\phi_n\rangle$,
      so $|\phi_n\rangle$ is a left eigenvector of $H$
      with the same eigenvalue $E_n$;
    \item $\langle\phi_m|\psi_n\rangle = \delta_{mn}$
      (biorthonormality);
    \item $\sum_n |\psi_n\rangle\langle\phi_n| = \mathbf{1}$
      (completeness);
    \item $\Tr_{\rm{bi}}[A] = \Tr_\eta[A]$
      for all $A \in \mathcal{B}(\mathcal{H})$.
  \end{enumerate}
\end{proposition}

\begin{proof}
\noindent\textbf{(i).}
Using the pseudo-Hermitian relation in the form $H^\dag\eta = \eta H$, one can get:
\[
  H^\dag|\phi_n\rangle
  = H^\dag(\eta|\psi_n\rangle)
  = \eta(H|\psi_n\rangle)
  = \eta(E_n|\psi_n\rangle)
  = E_n\,\eta|\psi_n\rangle
  = E_n|\phi_n\rangle.
\]

\noindent\textbf{(ii).}
By the definition $|\phi_m\rangle = \eta|\psi_m\rangle$ and
the $\eta$-orthonormality of the eigenbasis from \ref{asm:A3}, one can obtain:
\[
  \langle\phi_m|\psi_n\rangle
  = \langle\psi_m|\eta|\psi_n\rangle
  = \langle\psi_m|\psi_n\rangle_\eta
  = \delta_{mn}.
\]

\noindent\textbf{(iii).}
For any $|\chi\rangle \in \mathcal{H}$, expanding in the
$\eta$-orthonormal basis $\{|\psi_n\rangle\}$ of \ref{asm:A3}\ 
gives $|\chi\rangle = \sum_n \langle\psi_n|\chi\rangle_\eta\,|\psi_n\rangle$.
Using $\langle\psi_n|\chi\rangle_\eta = \langle\psi_n|\eta|\chi\rangle
= \langle\phi_n|\chi\rangle$, one can find:
\[
  \sum_n |\psi_n\rangle\langle\phi_n|\chi\rangle
  = \sum_n |\psi_n\rangle\langle\psi_n|\chi\rangle_\eta
  = |\chi\rangle.
\]
Since $|\chi\rangle$ is arbitrary, $\sum_n |\psi_n\rangle\langle\phi_n|
= \mathbf{1}$.

\noindent\textbf{(iv).}
Using part~(ii) and $|\phi_n\rangle = \eta|\psi_n\rangle$, one can give:
\[
  \Tr_{\rm{bi}}[A]
  = \sum_n \langle\phi_n|A|\psi_n\rangle
  = \sum_n \langle\psi_n|\eta A|\psi_n\rangle
  = \Tr_\eta[A]. \qedhere
\]
\end{proof}

\section{Equivalence of Routes~I and~II}
\label{app:equiv}

This appendix establishes the precise relationship between the
two routes. Under the full Assum.~(A1), the
biorthogonal state $\omega_{\rm{bi}}$ of Route~II coincides
with the $\eta$-Gibbs state $\omega_\eta$ of Route~I, so the
biorthogonal KMS-type identity (Prop.~\ref{prop:formkms})
automatically inherits all the properties established for
$\omega_\eta$ in Sec.~\ref{sec:route1}. In finite dimensions,
the converse is also true: by the Biorthogonal KMS Structure
Theorem (Thm.~\ref{thm:structure}), positivity of
$\omega_{\rm{bi}}$ implies quasi-Hermiticity, so
Assum.~\ref{asm:A1} is in fact \emph{equivalent} to positivity of
$\omega_{\rm{bi}}$ in the finite-dimensional setting.
The present appendix formalises the forward direction, valid in
any dimension, and records the status of the converse in
infinite dimensions as an open problem.

\begin{theorem}[Equivalence of Routes~I and~II under
                Assum.~(A1)]
\label{thm:equiv}
Under Assum.~\ref{asm:A1},
$\omega_{\rm{bi}} = \omega_\eta$.
Consequently:
\begin{enumerate}
\item The biorthogonal partition functions coincide:
      $Z_{\rm{bi}} = Z_\eta$.
\item The biorthogonal KMS-type identity
    (Prop.~\ref{prop:formkms}) upgrades to the full
    Spectral KMS condition: $\omega_{\rm{bi}}$ satisfies
    all three analytic conditions of Def.~\ref{def:KMS}.
\item $\omega_{\rm{bi}}$ is faithful and positive
      (Thm.~\ref{thm:positivity}).
\end{enumerate}
\end{theorem}

\begin{proof}
By Proposition~\ref{prop:biortho-struct}(iv),
$\Tr_{\rm{bi}} = \Tr_\eta$ under \textbf{(A1)}.
Therefore, for any $A \in \mathcal{B}(\mathcal{H})$:
\[
\omega_{\rm{bi}}(A)
=
\frac{1}{Z_{\rm{bi}}}
\sum_n \me^{-\beta E_n} A_{nn}
=
\frac{\Tr_\eta[\me^{-\beta H} A]}{Z_\eta}
=
\omega_\eta(A),
\]
where $Z_{\rm{bi}} = \sum_n \me^{-\beta E_n} = Z_\eta$ by \eqref{eq:partition}. 
The identity $\omega_{\rm{bi}} = \omega_\eta$
means that every property established for $\omega_\eta$ in
Sec.~\ref{sec:route1} holds equally for $\omega_{\rm{bi}}$,
in particular, the three KMS conditions of Thm.~\ref{thm:KMS} and the faithfulness and positivity of Thm.~\ref{thm:positivity}.
\end{proof}

\begin{remark}[Status of the converse in infinite dimensions]
\label{rem:converse}
In finite dimensions ($\mathcal{H} = \mathbb{C}^d$), the full equivalence
\[
\text{positivity of } \omega_{\rm{bi}}
\;\Longleftrightarrow\;
\text{quasi-Hermiticity of } H
\;\Longleftrightarrow\;
\text{Assumption~(A1--A4)}
\]
is established by Thm.~\ref{thm:structure}: in particular,
the implication (i)$\Rightarrow$(ii) is proved there, 
and no open problem remains in the finite-dimensional case.

In infinite dimensions \cite{Mostafazadeh:2012ezk}, Theorem~\ref{thm:equiv} above
establishes only the forward direction:
\[
\text{Assum.~(A1)}
\;\Longrightarrow\;
\omega_{\rm{bi}} \text{ is a faithful positive KMS state}
\;\left(= \omega_\eta\right).
\]
The converse — whether positivity of $\omega_{\rm{bi}}$
as a $*$-functional on $\mathcal{B}(\mathcal{H})$ implies the
existence of a positive-definite $\eta$ satisfying
(A1) when $\mathcal{H}$ is infinite-dimensional —
is related to the \emph{quasi-Hermitian inverse problem}
\cite{Scholtz:1992zz, Mostafazadeh:2008pw}. 
The difficulty is that the Gram-type construction used in the finite-dimensional proof (Steps~1--3 of Thm.~\ref{thm:structure}) requires
the representing matrix $G$ to define a bounded operator with bounded inverse, 
which is not automatic in infinite dimensions. 
This question is outside the scope of the present work and is left as an open problem.
\end{remark}

\bibliographystyle{utphys}
\bibliography{references}
\end{document}